\documentclass[11pt]{article}

\usepackage[final]{acl}
\usepackage{times}
\usepackage{latexsym}
\usepackage[T1]{fontenc}
\usepackage[utf8]{inputenc}
\usepackage{microtype}
\usepackage{inconsolata}
\usepackage{placeins}

\usepackage[skip=2pt]{caption}

\usepackage{wrapfig}
\usepackage{hyperref}
\usepackage{url}
\usepackage{booktabs}

\usepackage{lineno}

\definecolor{darkblue}{rgb}{0, 0, 0.5}
\hypersetup{colorlinks=true, citecolor=darkblue, linkcolor=darkblue, urlcolor=darkblue}

\usepackage{graphicx}
\usepackage{subcaption}
\usepackage{amsthm}
\usepackage{amsmath}
\usepackage{amssymb}
\usepackage{xspace}
\usepackage{bm}
\usepackage{multirow}
\usepackage{makecell}
\usepackage{xcolor}
\usepackage{tcolorbox}
\usepackage{float}
\usepackage{algorithm}
\usepackage{algpseudocode}
\newcommand{\myparatight}[1]{\smallskip\noindent{\bf {#1}:}~}
\newcommand{\alg}{\texttt{NeuronGuard}\xspace}

\newtheorem{assumption}{Assumption}
\newtheorem{theorem}{Theorem}
\newtheorem{lemma}{Lemma}
\newtheorem*{remark}{Remark}

\usepackage{color, colortbl}
\definecolor{greyL}{RGB}{229,229,255}

\title{NeuronGuard: Robust LLM Safety Alignment via Ablation-Aware Safety Signal Redistribution}

\author{Anjun Gao$^{1}$\quad Yueyang Quan$^{2}$ \quad Yufei Xia$^{1}$\thanks{Yufei Xia performed this research when he was under the supervision of Minghong Fang.} \quad Zhuqing Liu$^{2}$ \quad Minghong Fang$^{1}$ \\
$^1$University of Louisville,
$^2$University of North Texas
}

\begin{document}

\maketitle


\begin{abstract}

Safety alignment in large language models (LLMs) remains brittle against a growing spectrum of attacks. Jailbreak attacks bypass safety mechanisms through crafted prompts, while neuron-level attacks directly prune safety-critical neurons post-deployment. Both exploit a common weakness: safety-relevant information concentrates in a sparse neuron subset. We present \alg, a fine-tuning-stage defense that simultaneously hardens LLMs against both attack classes by redistributing safety signals across a broader set of neurons. \alg dynamically identifies safety-critical neurons via periodically refreshed per-layer linear classifiers, forces refusal behavior under deliberate neuron ablation, and applies KL-divergence regularization for distributional consistency. A randomized gradient projection strategy preserves downstream task utility by resolving conflicts between the defense and task objectives. We provide a formal guarantee that \alg strictly reduces the attack success rate (ASR) upper bound, and experiments across three LLMs, six state-of-the-art attack strategies, and multimodal settings confirm near-zero ASR while maintaining task accuracy, including against white-box adaptive adversaries.

\end{abstract}


\section{Introduction}
\label{sec:intro}

Large language models (LLMs)~\citep{qwen2.5,grattafiori2024llama,Falcon3,team2024gemma,abdin2024phi,brown2020language,qwen35blog,blakeman2025nemotron} have demonstrated remarkable capabilities across diverse applications, from question answering and code generation to complex reasoning and planning. Beyond these foundational tasks, LLMs have been increasingly deployed as general-purpose intelligence components powering autonomous agents, dialogue systems, and tool-augmented workflows. This broad and rapid adoption has made the safety and robustness of LLMs a critical concern, particularly given the potential consequences of deploying misaligned or manipulable models in high-stakes domains.

To ensure safe deployment, these models undergo extensive alignment procedures, including supervised fine-tuning~\citep{ouyang2022training}, reinforcement learning from human feedback~\citep{bai2022training}, and preference optimization~\citep{rafailov2023direct} to instill refusal behavior against harmful requests. However, this alignment has proven fragile across a growing spectrum of post-alignment attacks. 
Jailbreak attacks~\citep{chao2025jailbreaking,mehrotra2024tree,chang2024play,zou2023universal,liu2023autodan} bypass safety mechanisms through crafted prompts, adversarial suffixes, or role-playing scenarios without modifying model parameters. At a deeper level, neuron-level attacks~\citep{wu2025neurostrike} strike at the internal structure of aligned models: prior work has revealed that safety alignment tends to concentrate in a sparse set of specialized neurons, and NeuroStrike~\citep{wu2025neurostrike} exploits this by identifying and pruning safety-critical neurons at inference time, disabling safety mechanisms without any adversarial prompt. Both attack classes share a common root cause, namely that safety-relevant information concentrates in a sparse neuron subset, yet they exploit it from orthogonal directions: one through the input and one through the model's parameters.

Existing defenses are inadequate against this dual threat. Prompt-level defenses such as perplexity filtering~\citep{alon2023detecting}, SmoothLLM~\citep{robey2023smoothllm}, and GradSafe~\citep{xie2024gradsafe} are entirely ineffective against attacks that operate at the parameter level. 
Existing model-level defenses, including CAT~\citep{xhonneux2024efficient} and LED~\citep{zhao2024defending}, offer more direct protection but remain vulnerable to direct neuron manipulation.
SafeNeuron~\citep{wang2026safeneuron} comes closest by freezing safety-critical neurons during optimization, yet its \emph{static} identification strategy cannot account for representational shifts as model weights evolve. More fundamentally, none of these approaches addresses the root cause: safety information remains concentrated in a sparse neuron subset, leaving the model vulnerable to any adversary who can identify and suppress it.

\noindent \textbf{Our contributions: }%
To address these challenges, we present \alg, a fine-tuning-stage defense framework that hardens LLMs against both jailbreak and neuron-level attacks by redistributing safety signals across a broader set of neurons. Our method builds on a key observation: the brittleness of current alignment stems from over-concentration of safety-relevant information in a sparse neuron subset, and robustness can be improved by training the model to maintain safe behavior even when primary safety neurons are suppressed. To identify target neurons, \alg fits lightweight per-layer linear classifiers that separate harmful from benign representations at each fine-tuning step; these classifiers are periodically refreshed to track the \emph{current} safety-critical neuron set as weights evolve. Building on this, \alg constructs a safety objective that deliberately ablates the identified neurons and requires the model to refuse harmful queries in their absence, forcing gradient updates into the remaining neurons and progressively broadcasting safety representations across the network. A KL-divergence regularizer constrains the ablated model's output to match the standard forward pass, ensuring behavioral consistency. As safety signals spread to new neurons, those neurons are identified and ablated in turn, continuously expanding representational redundancy. Finally, since safety and task objectives frequently produce conflicting gradients, \alg applies a randomized projection that removes the opposing component from one gradient without systematically favoring either objective, preserving utility without sacrificing robustness.

We conduct evaluations across diverse attack settings, demonstrating that \alg effectively defends against both jailbreak attacks and neuron-level attacks while maintaining model utility on downstream fine-tuning tasks. We evaluate against six representative attack strategies, assess robustness in multimodal settings, and include challenging adaptive adversaries with full knowledge of our defense. Across all settings, \alg consistently reduces the attack success rate (ASR) to near zero while preserving task accuracy, with modest computational overhead suitable for practical fine-tuning pipelines.

Our contributions are summarized as follows:

\begin{list}{\labelitemi}{\leftmargin=1em \itemindent=0em \itemsep=0.3em \parsep=0em \topsep=0em \partopsep=0em}
	
	\item We propose \alg, a novel fine-tuning-stage defense framework that dynamically identifies safety-critical neurons via periodically refreshed per-layer linear classifiers, and broadcasts safety signals across a broader set of neurons through ablation-robust optimization and randomized gradient projection, providing simultaneous resistance to both jailbreak attacks and neuron-level attacks.
	
	\item We provide a formal theoretical guarantee showing that \alg strictly reduces the ASR upper bound relative to an undefended model under neuron-level attack, with the reduction quantified in terms of the ablation coverage ratio, per-neuron attack success probability, and the KL-divergence regularization strength.
	
	\item Comprehensive evaluations across diverse attack strategies and deployment scenarios demonstrate that \alg effectively hardens safety alignment while preserving model utility, including under adaptive adversaries explicitly designed to evade our defense.
	
\end{list}


\section{Related work} \label{sec:related}

\myparatight{Post-alignment attacks on LLMs}%
Post-alignment attacks are inference-time adversarial strategies that bypass safety alignment in deployed LLMs. We focus on two representative families: jailbreak attacks and neuron-level attacks.
Jailbreak attacks subvert safety alignment through crafted inputs without modifying model parameters. Generation-based methods such as PAIR~\citep{chao2025jailbreaking}, TAP~\citep{mehrotra2024tree}, and Puzzler~\citep{chang2024play} synthesize adversarial prompts via iterative querying or puzzle-like reformulations. Optimization-based methods take a more direct route: GCG~\citep{zou2023universal} searches over adversarial token suffixes via gradient guidance, while AutoDAN~\citep{liu2023autodan} combines gradient signals with fluent prompt generation.
Rather than manipulating inputs, neuron-level attacks exploit the internal structure of aligned models directly. 
Safety alignment is concentrated in a sparse set of specialized neurons~\citep{wu2025neurostrike}, making it both localizable and fragile. 
NeuroStrike~\citep{wu2025neurostrike} compromises an aligned model by locating neurons responsible for refusal behavior and suppressing their activations during inference, thereby inducing unsafe outputs without requiring specially crafted inputs.

\myparatight{Defenses against post-alignment attacks}%
Existing defenses fall into two categories, each with fundamental limitations. \emph{Prompt-level defenses} intercept adversarial inputs at inference time without altering the model: Perplexity filtering~\citep{alon2023detecting,cheng2025secure} flags anomalous inputs, SmoothLLM~\citep{robey2023smoothllm} aggregates responses over randomly perturbed inputs via majority voting, and GradSafe~\citep{xie2024gradsafe} detects jailbreak attempts through gradient patterns on safety-critical parameters. However, these defenses are inherently ineffective against neuron-level attacks, which operate at the parameter level and require no adversarial prompt.
\emph{Model-level defenses} modify the model's parameters or training procedure, offering stronger but still insufficient protection. CAT~\citep{xhonneux2024efficient} performs adversarial training in the continuous embedding space, and LED~\citep{zhao2024defending} edits safety-critical layers via hidden-state activation analysis; both target prompt-based threats and offer little resistance to direct neuron manipulation. SafeNeuron~\citep{wang2026safeneuron} freezes safety-critical neurons during direct preference optimization to encourage redundant safety pathways, but its static identification strategy cannot capture representation shifts during weight updates. Fundamentally, none of these approaches trains the model to distribute safety signals broadly, which is the root cause of their fragility.

Note that a complementary line of research has investigated post-attack forensic attribution, which aims to determine how a successful attack was carried out and identify the source or root cause responsible for its success~\citep{jia2024tracing,gao2026patching,zhang2025traceback,zhang2026Who,gao2026beware}. These techniques provide valuable information for diagnosing compromised systems and understanding the origins of successful attacks, thereby complementing conventional defense mechanisms. Nevertheless, post-attack forensics addresses a different stage of the attack lifecycle from the mitigation strategies considered in this work. Our study instead focuses on preventing and mitigating attacks, while forensic investigation and attribution after a successful attack are beyond the scope of this paper.


\section{Threat model} 
\label{sec:problem}

\myparatight{Attacker's goal and knowledge}%
We consider post-alignment attacks that aim to bypass a model’s safety alignment and elicit harmful responses.
This framing is consistent with both jailbreak attacks~\citep{chao2025jailbreaking,mehrotra2024tree,chang2024play,zou2023universal,liu2023autodan} and neuron-level attacks~\citep{wu2025neurostrike}, which are inherently post-hoc. We further assume a worst-case setting where the attacker has full access to the post-finetuning weights and can manipulate internal parameters or neurons, a realistic assumption for open-weight LLMs whose parameters are publicly available.

\myparatight{Defender's goal and knowledge}%
The defender aims to maintain safety alignment under post-alignment attacks, ensuring consistent refusal of harmful queries.
We assume the defender operates during fine-tuning with full access to the training pipeline, including data, weights, and objectives. This is natural: the developer who conducts fine-tuning is directly responsible for the safety properties of the released model. Critically, fine-tuning is the last stage under the developer's control before public deployment, as once released, the developer cannot control how weights are accessed or modified. Hardening the model at this stage therefore proactively reduces the attack surface, without requiring any intervention at inference time or any assumptions about the attacker's specific strategy. Following prior work~\citep{wu2025neurostrike,wang2026safeneuron}, we assume the defender holds a small safety probe set of harmful and benign queries.
This assumption is reasonable as any responsible model developer would already maintain a small set of safety-critical queries as part of standard quality control before deployment, and such a probe set requires no knowledge of the attacker's specific strategy.


\section{Our method}
\label{our_method}

\alg addresses the root cause of alignment fragility through three coupled stages, as illustrated in Fig.~\ref{fig:neuronguard} (Appendix). The first stage fits per-layer linear classifiers to identify the sparse neuron subset dominating safety-gating decisions, refreshed periodically as weights evolve. The second stage deliberately suppresses these critical neurons during training, forcing safety representations to redistribute across a broader population so that no small subset serves as a single point of failure. The third stage resolves conflicts between the defense and task objectives via randomized gradient projection, preserving utility without sacrificing robustness. The complete procedure is summarized in Algorithm~\ref{alg:main} (Appendix).

\subsection{Stage 1: Identifying safety-critical neurons}
\label{subsec:identify}

Safety signals must be redistributed away from the neurons that currently dominate the
model's refusal behavior. 
The first procedure is therefore to identify those neurons precisely
at each point in fine-tuning. 
We achieve this by fitting a lightweight linear classifier at each layer
that separates harmful from benign representations, and then reading off which neurons
contribute most to that separation.

Concretely, the defender maintains a safety probe set consisting of two small subsets: a harmful query set $D_\text{unsafe}$ and a benign query set $D_\text{safe}$.
Each prompt $x \in D_\text{safe} \cup D_\text{unsafe}$
is assigned a binary label ($y=1$ for harmful, $y=0$ for benign). 
We then feed $x$ through the current model and collect the activation vectors $h_l(x)$ at each layer $l$, and a per-layer logistic classifier is trained on these activations:
\begin{equation}
	\label{eq:classifier}
	\hat{y}_l(x) = \sigma (W_l^\top h_l(x) + b_l), \quad \forall\, x \in D_\text{safe} \cup D_\text{unsafe},
\end{equation}
where $\sigma(\cdot)$ is the sigmoid function and $W_l, b_l$ are the learnable classifier
parameters. The weight vector $W_l$ encodes neuron-level importance: a large positive
weight $W_{l,i}$ indicates that neuron $i$ is a strong predictor of the ``unsafe'' label, and
hence a primary carrier of safety-relevant information. We therefore rank neurons by their
corresponding classifier weights and select the top-$\rho$ fraction as the safety-critical set for
layer $l$:
\begin{equation}\label{eq:topfrac}
	\mathcal{I}_{\text{safety}, l}
	= \operatorname{TopFrac}\!\left(W_l,\; \rho\right),
\end{equation}
where $\operatorname{TopFrac}(\cdot, \rho)$ returns the indices of the largest $\rho \cdot d_l$
elements, and $d_l$ is the number of neurons in layer $l$.
Because model representations shift as weights are updated, a fixed neuron set identified
at initialization quickly becomes stale. We therefore refresh the classifiers, and
recompute $\mathcal{I}_{\text{safety}, l}$, every $N$ optimization steps. This periodic
re-identification ensures that Stage 2 always targets the \emph{current} safety-critical
neurons, not an outdated approximation, and is the mechanism that produces the iterative
broadening effect described next.

\subsection{Stage 2: Ablation-robust safety optimization}
\label{sec:optimization}

With the current safety-critical neuron sets $\{\mathcal{I}_{\text{safety}, l}\}$ in hand, we
construct a safety objective that drives the model to spread safety representations beyond
those neurons. The objective has three terms, each enforcing a distinct aspect of robust
safety.

\myparatight{Refusal loss under normal execution}%
The first term ensures the model refuses harmful queries during normal inference. Let
$x_{\text{h}} \in D_\text{unsafe}$ denote a harmful input and $y_{\text{ref}}$ a fixed
refusal-template target (e.g., ``I'm sorry, but I cannot assist with that request.''). We
apply standard cross-entropy supervision:
\begin{equation}\label{eq:lr}
	L_{\text{ref}} = \mathrm{CE}\!\left(y_{\text{ref}},\; p_{\theta}(y \mid x_{\text{h}})\right),
\end{equation}
where $p_{\theta}(y \mid x_{\text{h}})$ is the model's output distribution under parameters
$\theta$. This term provides a baseline safety signal, but on its own, only reinforces the
neurons that already encode refusal behavior. Safety information, therefore, remains
concentrated in $\mathcal{I}_{\text{safety}, l}$, leaving the model vulnerable whenever
those neurons are suppressed. The next term addresses this limitation directly.

\myparatight{Refusal loss under safety-neuron ablation}%
To compel other neurons to take on safety responsibilities, we introduce a second
supervision term that operates on a \emph{modified} forward pass in which the
safety-critical neurons are deliberately zeroed out. Specifically, for each layer $l$, we
construct ablated activations $\tilde{h}_{i,l} = 0$ if $i \in \mathcal{I}_{\text{safety},l}$,
and $\tilde{h}_{i,l} = h_{i,l}$ otherwise, then propagate them through the remaining
layers to obtain the ablated output distribution $p_{\theta_{\text{abl}}}(y \mid x_{\text{h}})$.
We then apply refusal supervision on this ablated forward pass:
\begin{equation}\label{eq:lr_abl}
	L_{\text{ref}}^{\text{abl}} = \mathrm{CE}\!\left(y_{\text{ref}},\; p_{\theta_{\text{abl}}}(y \mid x_{\text{h}})\right).
\end{equation}
Because the neurons in $\mathcal{I}_{\text{safety}, l}$ are zeroed out, gradients from
$L_{\text{ref}}^{\text{abl}}$ cannot flow back into them and must instead propagate to the
remaining neurons, training them to carry safety-relevant information. As these neurons
absorb safety information over successive steps, the refreshed classifiers identify them
as part of the new safety-critical set and ablate them in turn, pushing the safety signal
progressively outward to an ever-wider set of neurons.

\myparatight{Distributional consistency regularization}%
Optimizing $L_{\text{ref}}$ and $L_{\text{ref}}^{\text{abl}}$ in isolation does not prevent the ablated
model from producing correct refusal tokens while diverging sharply on the remaining
vocabulary, leading to unstable training dynamics. To prevent this, we add a
KL-divergence regularizer between the two forward passes on the same harmful input:
\begin{equation}\label{eq:kl}
	L_{\text{cons}} = D_{\text{KL}}(p_{\theta_{\text{abl}}}(y \mid x_{\text{h}}) \;\|\; p_{\theta}(y \mid x_{\text{h}})).
\end{equation}

This term enforces that ablating safety neurons preserves the model’s overall output distribution, not just the refusal token, preventing degenerate collapse.

\begin{figure}[t]
    \centering
    \includegraphics[width=0.45\textwidth]{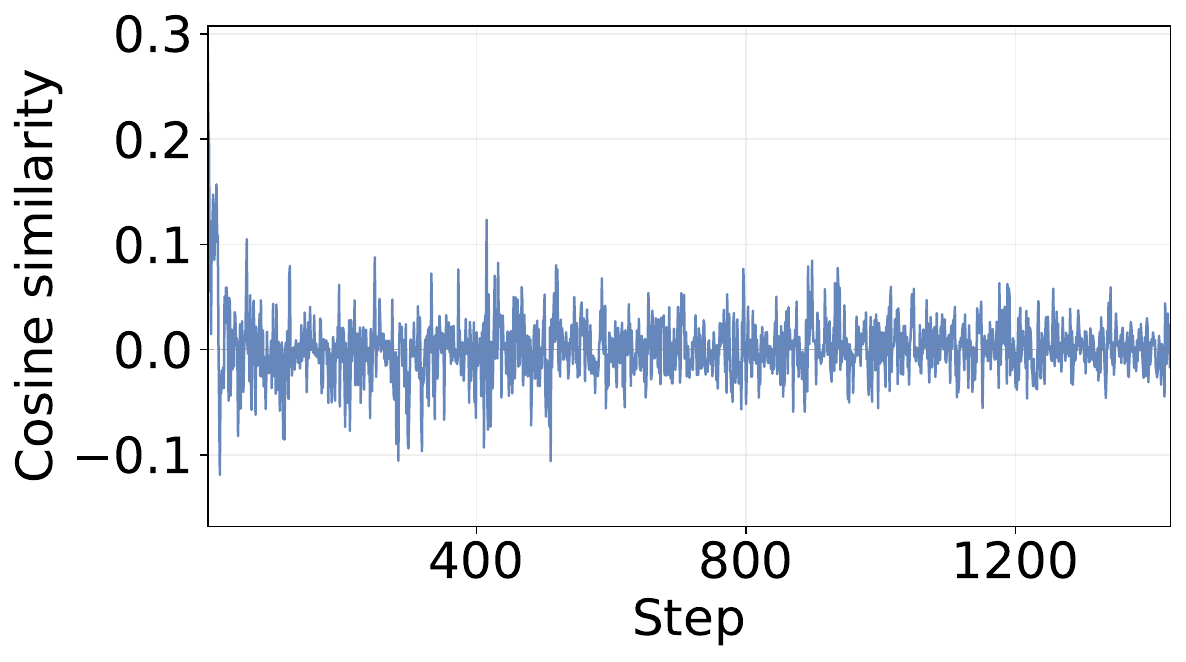}
    \caption{Cosine similarity between $g_{\text{user}}$ and $g_{\text{safe}}$ across fine-tuning steps on SST2~\citep{socher2013recursive} with Llama.}
    \label{fig:grad-sign-pos-neg}
	\vspace{-0.2in}
\end{figure}

\myparatight{Full safety objective}%
Combining the three terms, the complete safety objective is:
\begin{equation}\label{eq:defense}
	L_{\text{safe}} = L_{\text{ref}} + L_{\text{ref}}^{\text{abl}} + L_{\text{cons}}.
\end{equation}
$L_{\text{ref}}$ ensures that the model refuses harmful queries in standard setting; $L_{\text{ref}}^{\text{abl}}$
forces safety representations to spread beyond the identified safety neurons; and
$L_{\text{cons}}$ maintains distributional consistency between the two settings.
Together, these three terms promote distributed safety encoding across neurons, avoiding concentration in only a few neurons.

\subsection{Stage 3: Resolving gradient conflicts via randomized projection}

Beyond the safety objective, the model must simultaneously learn the user's downstream
task. Let $L_{\text{user}}$ denote the task loss computed on the user-provided fine-tuning
dataset $D_{\text{user}}$. The overall fine-tuning objective is:
\begin{equation}\label{eq:total}
	L_{\text{total}} = L_{\text{user}} + L_{\text{safe}},
\end{equation}
where $L_{\text{safe}}$ is the safety objective defined in Eq.~\eqref{eq:defense}.
However, naively minimizing $L_{\text{total}}$ with a single gradient step is problematic
in practice. Let $g_{\text{user}} = \nabla_{\theta} L_{\text{user}}$ and
$g_{\text{safe}} = \nabla_{\theta} L_{\text{safe}}$ denote the gradients of the two
objectives. As shown in Figure~\ref{fig:grad-sign-pos-neg}, in a large fraction of fine-tuning steps we observe $\langle g_{\text{user}},\, g_{\text{safe}} \rangle < 0$, meaning the two
gradients point in conflicting directions. Directly summing them would therefore cause
destructive interference: progress on one objective would come at the expense of the
other, ultimately harming task utility.

To address this, we introduce a randomized gradient correction that resolves conflicts
at each fine-tuning step without systematically favoring either objective. At each step, we
first compute $\langle g_{\text{user}},\, g_{\text{safe}} \rangle$. If no conflict is
detected, both gradients are used unchanged. If a conflict is detected, we randomly
sample an ordering $(g_1, g_2)$ of $(g_{\text{user}}, g_{\text{safe}})$, where $g_1$ is
the gradient to be corrected and $g_2$ is kept unchanged. The corrected gradient is:
\begin{equation}\label{eq:projection}
	\hat{g}_1 = g_1 - \frac{\langle g_1,\, g_2 \rangle}{\| g_2 \|_2^2}\, g_2, \qquad
	\hat{g}_2 = g_2,
\end{equation}
which removes from $g_1$ precisely the component that opposes $g_2$, while leaving all
non-conflicting components intact. The projection is applied to only one gradient rather
than both, to avoid unnecessarily shrinking the effective update magnitude. Randomizing
which gradient serves as $g_1$ ensures that neither objective is systematically
deprioritized across fine-tuning. When no conflict is detected, we simply set
$\hat{g}_1 = g_{\text{user}}$ and $\hat{g}_2 = g_{\text{safe}}$.

The final gradient update is then:
\begin{equation}\label{eq:final_grad}
	g_{\text{final}} = \hat{g}_1 + \hat{g}_2.
\end{equation}
This ensures that the two objectives always make progress in mutually non-opposing
directions, preserving both task utility and the safety robustness gained in Stage~2.
The model parameters are then updated as $\theta \leftarrow \theta - \gamma \cdot g_{\text{final}}$,
where $\gamma$ is the learning rate.


\section{Theoretical analysis} 
\label{theoretical_analysis}

We provide theoretical guarantees for our proposed \alg under neuron-level attack
(which is more powerful than jailbreak attacks based on our experimental observations),
answering whether safety-neuron ablation provably reduces ASR and under what conditions.
Guarantees are stated in terms of measurable quantities (classifier accuracy, ablation
coverage, and KL divergence) estimable from fine-tuning logs. Let $p_\theta(\cdot\mid x)$
denote the model's output distribution and $h_{i,l}(x)$ the activation of neuron $i$ at
layer $l$. The ground-truth safety-neuron set at layer $l$ is $S_{l}=\{\,i \mid
\mathbb{E}_{x\sim\mathcal{D}_{\text{unsafe}}}[h_{i,l}(x)]-
\mathbb{E}_{x\sim\mathcal{D}_{\text{safe}}}[h_{i,l}(x)]>\tau,\; i\in[1,d_l]\,\}$,
where $\tau>0$ and $|S_l|=s_l\ll d_l$. The ablation coverage ratio is defined as
$\rho_l := |S_l \cap \mathcal{I}_{\text{safety},l}|/|S_l|$. We first provide the necessary
assumptions.

\begin{assumption}[Approximate linear separability and sufficient ablation coverage]
\label{asmp1}
There exists a linear classifier $\psi=(W_l,b_l)$ with
$\Upsilon(h_l(x);\psi)=\mathbf{1}\{W_l^\top h_l(x)+b_l\ge 0\}$
achieving small misclassification rates $\delta_u,\delta_b\in(0,1)$ on unsafe and benign queries respectively, where $\Upsilon(\cdot;\psi)$ is the binary decision rule of the linear classifier $\hat{y}_l(x)$ in Eq.~\eqref{eq:classifier}.
The ablation coverage satisfies $\rho_l \ge \rho_0$ for some constant $\rho_0\in(0,1)$.
\end{assumption}

\begin{assumption}[Bounded attack capability]
\label{asmp2}
The attacker perturbs only neurons in $S_l$, with $\|\delta\|_2\le\varepsilon$.
After ablation, the effective perturbation support is restricted to $S_l\setminus\mathcal{I}_{\text{safety},l}$.
\end{assumption}

\begin{assumption}[Local Lipschitz continuity of the output head]
\label{asmp3}
The subnetwork from layer~$l$ to the output, $z_\theta(\cdot)$, satisfies
$\|z_\theta(h_l)-z_\theta(h_l')\|_2 \le C_L\|h_l-h_l'\|_2$
in a neighborhood of $h_l(x)$, and all admissible perturbations remain within this neighborhood, where $C_L > 0$ is a constant.
\end{assumption}

\begin{assumption}[Bounded KL divergence between ablated and standard outputs]
\label{asmp4}
At convergence, $\mathbb{E}_{x\sim\mathcal{D}_{\text{unsafe}}}[ D_{\mathrm{KL}}(
p_{\theta_{\text{abl}}}(y\mid x)\|p_{\theta}(y\mid x))] \le \eta$ for a finite constant $\eta>0$.
\end{assumption}

\begin{assumption}[Independent per-neuron exploitation]
\label{asmp5}
Each non-ablated safety neuron is independently exploitable with probability $\alpha_l\in[0,1]$.
\end{assumption}

Extended discussion and verifiability remarks are provided in Appendix~\ref{sec:appendix_remarks}.
We then state the main result showing that ablating safety neurons provably reduces the
ASR upper bound relative to an undefended model, with the proof relying on three auxiliary
lemmas in Appendix~\ref{app:proof_lemma}: Lemma~\ref{lemma1} shows ablation shrinks the
adversarial subspace; Lemma~\ref{lemma2} provides Lipschitz control of output perturbations;
and Lemma~\ref{lemma3} translates bounded KL divergence into bounded $\ell_1$ distance
via Pinsker's inequality.

\begin{tcolorbox}[
  width=1\linewidth,
  left=0.5mm,
  right=0.5mm,
  top=0.5mm,
  bottom=0.5mm,
  before upper={%
    \setlength{\leftskip}{0pt}%
    \setlength{\rightskip}{0pt}%
  }
]
\begin{theorem}[ASR upper-bound reduction]
\label{theorem1}

Suppose Assumptions~\ref{asmp1}--\ref{asmp5} hold.
Let $s_l = |S_l|$ denote the number of safety neurons at layer $l$,
$\rho_l\in(0,1]$ the ablation coverage ratio, and $\alpha_l\in(0,1)$
the per-neuron attack success probability.
Define the neuron-level exploitation probabilities as
$P_{\mathrm{nl}} = 1 - (1-\alpha_l)^{s_l}$ and the exploitation probability under our defense as
$P_{\mathrm{ours}} = 1 - (1-\alpha_l)^{(1-\rho_l)s_l}$. Let $\mathcal{B}(\mathrm{ASR})$ denote the theoretical upper bound of the attack success rate.
If\;
$\eta \le \big((1-\alpha_l)^{(1-\rho_l)s_l} - (1-\alpha_l)^{s_l}\big)^2/2$,\;
then $\mathcal{B}(\mathrm{ASR}_{\mathrm{nl}}) \ge \mathcal{B}(\mathrm{ASR}_{\mathrm{ours}})$,
i.e., \alg strictly reduces the ASR upper bound relative to the undefended model.

\end{theorem}
\end{tcolorbox}

The sufficient condition on $\eta$ formalizes an intuitive trade-off: ablation reduces the
attacker's exploitable subspace, but introduces a distributional shift between the ablated
and standard models. When KL regularization is tight enough, the defense's ASR upper bound
is strictly lower than that of the undefended model (proof in Appendix~\ref{app:proof_theorem}).


\section{Experiments} \label{sec:exp}

\begin{table*}[t]
\centering
\footnotesize
\addtolength{\tabcolsep}{-1pt}
\begin{tabular}{l l | cccccccc}
\toprule
Model & Task & No defense & Perplexity & SmoothLLM & GradSafe & CAT & LED & SafeNeuron & \cellcolor{greyL}\alg \\
\midrule
\multirow{4}{*}{Llama}
& SST2 & 0.93 & 0.93 & 0.83 & 0.91 & 0.86 & 0.88 & 0.87 & \cellcolor{greyL}0.92 \\
& AGNews & 0.91 & 0.91 & 0.82 & 0.88 & 0.84 & 0.87 & 0.82 & \cellcolor{greyL}0.90 \\
& CoLA & 0.82 & 0.82 & 0.74 & 0.80 & 0.76 & 0.75 & 0.76 & \cellcolor{greyL}0.79 \\
& GSM8K & 0.65 & 0.65 & 0.55 & 0.63 & 0.58 & 0.61 & 0.60 & \cellcolor{greyL}0.62 \\
\midrule
\multirow{4}{*}{Qwen}
& SST2 & 0.92 & 0.92 & 0.85 & 0.88 & 0.85 & 0.86 & 0.86 & \cellcolor{greyL}0.89 \\
& AGNews & 0.90 & 0.90 & 0.84 & 0.87 & 0.83 & 0.85 & 0.81 & \cellcolor{greyL}0.88 \\
& CoLA & 0.82 & 0.82 & 0.74 & 0.80 & 0.76 & 0.73 & 0.76 & \cellcolor{greyL}0.78 \\
& GSM8K & 0.65 & 0.65 & 0.46 & 0.62 & 0.58 & 0.55 & 0.60 & \cellcolor{greyL}0.58 \\
\midrule
\multirow{4}{*}{Falcon}
& SST2 & 0.94 & 0.94 & 0.87 & 0.92 & 0.87 & 0.88 & 0.88 & \cellcolor{greyL}0.93 \\
& AGNews & 0.91 & 0.91 & 0.84 & 0.90 & 0.84 & 0.86 & 0.82 & \cellcolor{greyL}0.88 \\
& CoLA & 0.83 & 0.83 & 0.76 & 0.77 & 0.77 & 0.75 & 0.77 & \cellcolor{greyL}0.78 \\
& GSM8K & 0.68 & 0.68 & 0.57 & 0.61 & 0.61 & 0.62 & 0.63 & \cellcolor{greyL}0.66 \\
\bottomrule
\end{tabular}
\caption{ACC$\uparrow$ of different methods across three models and four fine-tuning tasks.}
\label{tab:utility_acc}
\end{table*}

\begin{table*}[t]
\centering
\footnotesize
\addtolength{\tabcolsep}{-1pt}
\begin{tabular}{l l | cccccccc}
\toprule
Model & Attack & No defense & Perplexity & SmoothLLM & GradSafe & CAT & LED & SafeNeuron & \cellcolor{greyL}\alg \\
\midrule
\multirow{6}{*}{Llama}
& PAIR & 0.17 & 0.14 & 0.15 & 0.09 & 0.10 & 0.07 & 0.04 & \cellcolor{greyL}0.01 \\
& TAP & 0.15 & 0.13 & 0.13 & 0.08 & 0.12 & 0.09 & 0.05 & \cellcolor{greyL}0.00 \\
& Puzzler & 0.29 & 0.25 & 0.23 & 0.12 & 0.17 & 0.13 & 0.07 & \cellcolor{greyL}0.01 \\
& GCG & 0.35 & 0.16 & 0.27 & 0.15 & 0.09 & 0.05 & 0.03 & \cellcolor{greyL}0.01 \\
& AD & 0.31 & 0.28 & 0.24 & 0.11 & 0.13 & 0.09 & 0.05 & \cellcolor{greyL}0.01 \\
& NS & 0.89 & 0.64 & 0.83 & 0.72 & 0.58 & 0.49 & 0.22 & \cellcolor{greyL}0.04 \\
\midrule
\multirow{6}{*}{Qwen}
& PAIR & 0.23 & 0.18 & 0.19 & 0.14 & 0.18 & 0.14 & 0.06 & \cellcolor{greyL}0.00 \\
& TAP & 0.25 & 0.21 & 0.10 & 0.16 & 0.17 & 0.13 & 0.07 & \cellcolor{greyL}0.01 \\
& Puzzler & 0.41 & 0.37 & 0.39 & 0.25 & 0.26 & 0.20 & 0.10 & \cellcolor{greyL}0.01 \\
& GCG & 0.40 & 0.18 & 0.33 & 0.17 & 0.13 & 0.08 & 0.04 & \cellcolor{greyL}0.01 \\
& AD & 0.36 & 0.36 & 0.23 & 0.18 & 0.20 & 0.15 & 0.07 & \cellcolor{greyL}0.01 \\
& NS & 0.83 & 0.57 & 0.68 & 0.68 & 0.53 & 0.42 & 0.21 & \cellcolor{greyL}0.02 \\
\midrule
\multirow{6}{*}{Falcon}
& PAIR & 0.15 & 0.12 & 0.12 & 0.07 & 0.05 & 0.02 & 0.03 & \cellcolor{greyL}0.01 \\
& TAP & 0.16 & 0.15 & 0.15 & 0.08 & 0.07 & 0.04 & 0.02 & \cellcolor{greyL}0.00 \\
& Puzzler & 0.30 & 0.25 & 0.27 & 0.09 & 0.09 & 0.05 & 0.03 & \cellcolor{greyL}0.01 \\
& GCG & 0.28 & 0.11 & 0.24 & 0.13 & 0.16 & 0.12 & 0.06 & \cellcolor{greyL}0.01 \\
& AD & 0.22 & 0.20 & 0.16 & 0.11 & 0.13 & 0.09 & 0.05 & \cellcolor{greyL}0.01 \\
& NS & 0.81 & 0.71 & 0.74 & 0.59 & 0.66 & 0.57 & 0.28 & \cellcolor{greyL}0.02 \\
\bottomrule
\end{tabular}
\caption{ASR$\downarrow$ of different methods under six attack types across three models on the SST2 task.}
\label{tab:jailbreak_asr}
\vspace{-.2in}
\end{table*}

\subsection{Experimental setup}

\myparatight{Models, datasets, and evaluation metrics}%
We conduct experiments on three pre-trained LLMs: Llama-3.1-8B-Instruct~\citep{grattafiori2024llama}, Qwen2.5-7B-Instruct~\citep{qwen2.5}, and Falcon3-7B-Instruct~\citep{Falcon3}, and fine-tune LLMs on four datasets: SST2~\citep{socher2013recursive}, AGNews~\citep{zhang2015character}, CoLA~\citep{warstadt2019neural}, and GSM8K~\citep{cobbe2021GSM8K}. We assess each defense using two complementary metrics. \textbf{Utility accuracy (ACC)} measures task performance after fine-tuning, where a higher ACC indicates better preservation of model utility. \textbf{Attack success rate (ASR)} measures the fraction of unsafe responses produced under harmful queries, where a lower ASR indicates more effective safety preservation. Following prior work~\citep{inan2023llama,wu2025neurostrike}, we employ Llama-Guard-3-8B~\citep{grattafiori2024llama} as the safety judge. We report the license of all models and datasets we used in this paper in Appendix~\ref{app:licenses}.

\myparatight{Baselines and attacks}%
We compare our method against six defense baselines. Three are prompt-level defenses: Perplexity Filter~\citep{alon2023detecting}, SmoothLLM~\citep{robey2023smoothllm}, and GradSafe~\citep{xie2024gradsafe}. Three are model-level defenses: CAT~\citep{xhonneux2024efficient}, LED~\citep{zhao2024defending}, and SafeNeuron~\citep{wang2026safeneuron}. Detailed descriptions are provided in Appendix~\ref{sec:appendix_baseline}. We evaluate all defenses against six representative attacks spanning both attack families discussed in Section~\ref{sec:problem}. PAIR~\citep{chao2025jailbreaking}, TAP~\citep{mehrotra2024tree}, and Puzzler~\citep{chang2024play} are generation-based jailbreak attacks that synthesize adversarial prompts to bypass safety mechanisms. GCG~\citep{zou2023universal} and AutoDAN~\citep{liu2023autodan} are optimization-based jailbreak attacks that use gradient-based methods to craft token sequences that maximize the likelihood of unsafe responses. NeuroStrike~\citep{wu2025neurostrike} is a neuron-level attack that selectively prunes safety-critical neurons to disable safety mechanisms.

\myparatight{Parameter settings}%
We adopt LoRA~\citep{hu2022lora} for fine-tuning, with rank $r = 32$ and scaling factor $\alpha = 64$, fine-tuning for 10 epochs on four utility tasks each containing 2,000 samples. The learning rate $\gamma$ across all experiments is set to $10^{-5}$. 
For all baselines and attacks, we use the default hyperparameter settings reported in the original works. For our method, the fraction of masked neurons $\rho$ is set to 0.05, and the refresh interval $N$ for the linear classifier is set to 200 steps. The safety probe set consists of 750 harmful queries from $D_{\text{unsafe}}$ and 750 benign queries from $D_{\text{safe}}$, both sampled from BeaverTails~\citep{ji2023beavertails}. 
We use StrongREJECT~\citep{souly2024strongreject} as the test dataset to compute ASR. All experimental results are reported as averages over 10 different random seeds, using two NVIDIA H100 GPUs, each with 94GB of memory.

\subsection{Experimental results}

\myparatight{\alg outperforms baselines}%
Table~\ref{tab:utility_acc} reports the ACC of fine-tuned models on four tasks under various defenses. Our proposed \alg preserves high ACC and closely matches the ``No defense'' setting, indicating strong utility without performance degradation. For instance, on Falcon (SST2), \alg achieves 0.93 ACC versus 0.94 under No defense, and on Llama (AGNews), it attains 0.90 versus 0.91.
Table~\ref{tab:jailbreak_asr} reports the ASR under six attacks on SST2 (AutoDAN and NeuroStrike are abbreviated as AD and NS); Table~\ref{tab:jailbreak_asr_agnews} to Table~\ref{tab:jailbreak_asr_gsm8k} (Appendix) report the ASR for AGNews, CoLA, and GSM8K respectively. Existing baselines fail to provide effective defense under adversarial conditions. 
In contrast,
\alg reduces ASR to nearly zero across all attacks, achieving 0.00 under PAIR on Qwen and 0.04 under NeuroStrike on Llama. Overall, \alg effectively defends against diverse attacks while maintaining high ACC on benign tasks.

\myparatight{Impact of $\rho$}%
We investigate how the selection ratio $\rho$ in the top-$\rho$ criterion, which controls the fraction of ablated safety-critical neurons, affects \alg. Table~\ref{tab:topk} (Appendix) reports ACC on SST2 and ASR under different attacks for different $\rho$. A larger $\rho$ zeros out more safety neurons, compelling the remaining neurons to learn safety-relevant representations more intensively. Overall, \alg exhibits stable and robust performance even at low $\rho$, demonstrating the effectiveness of our approach in broadcasting safety signals.

\myparatight{Impact of $N$}%
We investigate the effect of the refresh interval $N$, which controls how frequently safety-critical neurons are re-identified. Table~\ref{tab:refresh_N} (Appendix) reports ACC on SST2, ASR, and fine-tuning time for varying $N$. A smaller $N$ keeps neuron estimates accurate but incurs higher overhead, while a larger $N$ improves efficiency but may rely on outdated estimates. Overall, \alg maintains robust performance across a wide range of $N$, demonstrating that periodic re-identification effectively balances efficiency and accuracy.

\myparatight{Impact of the size of $D_\text{safe}$ and $D_\text{unsafe}$}%
We investigate the impact of the sizes of $D_\text{safe}$ and $D_\text{unsafe}$, used to identify safety-critical neurons and train refusal responses. Table~\ref{tab:safe_validation} (Appendix) reports ASR and ACC for dataset sizes ranging from 500 to 2,500 samples (set equally). Defense performance marginally improves as size increases, eventually reducing ASR to near zero. Notably, \alg achieves highly effective defense with as few as 500 samples, demonstrating that our method is data-efficient and requires only limited safety data to robustly broadcast safety signals.

\myparatight{Computation cost of \alg}%
Figure~\ref{fig:computational_overhead} (Appendix) reports the running time of each method on SST2 under NeuroStrike, 
while test-time methods (e.g., Perplexity) report only inference time. The overhead of \alg remains on the same order of magnitude as ``No defense''. Since \alg operates during fine-tuning, it incurs no additional cost at inference time, unlike test-time defenses that require extra operations per input, making \alg particularly well suited to efficiency-critical deployment scenarios.

\section{Discussion} 
\label{sec:discussion_limitation}

\myparatight{Different variants of \alg}%
To evaluate each component's contribution, we compare \alg against four variants: \textbf{Variant I} uses only refusal training $L_{\text{ref}}$ without neuron ablation; \textbf{Variant II} includes neuron ablation but omits the consistency regularizer $L_{\text{cons}}$; \textbf{Variant III} minimizes $L_{\text{user}} + L_{\text{safe}}$ without gradient projection; and \textbf{Variant IV} projects two gradients simultaneously in each step rather than using randomized projection. Table~\ref{tab:variants} (Appendix) reports ASR and ACC on SST2 across different attacks. Removing any component leads to either higher ASR or lower ACC, while the full \alg achieves the lowest ASR with high task accuracy, demonstrating that all components are indispensable.

\myparatight{Adaptive attacks}%
To evaluate worst-case resilience, we construct two adaptive attacks where the adversary has full knowledge of \alg. The first, iterative pruning (IP), repeatedly identifies and removes safety-related neurons on the progressively pruned model to eliminate newly emerging safety mechanisms. The second, nonlinear evasion (NE), leverages nonlinear models to uncover and remove safety-relevant neurons beyond those detectable by linear methods. More details are provided in Appendix~\ref{app:Adaptive_Attacks}. Table~\ref{tab:adaptive_attack} (Appendix) reports ASR under these attacks on SST2 (ACC remains identical to Table~\ref{tab:utility_acc} on Llama).
While existing defenses are highly vulnerable, \alg remains robust.
Further analysis of \alg's robustness under adaptive attacks is presented in Appendix~\ref{app:Adaptive_Attacks_analy}.

\myparatight{\alg is effective for the multimodal scenario}%
Safety is also critical for vision language models (VLMs). To assess \alg in multimodal settings, we evaluate ACC and ASR under image inputs on Qwen2.5-VL-7B-Instruct~\citep{qwen2.5-VL}. For the fine-tuning task, we use SST2 in text-to-image (T2I) format, with both $D_\text{safe}$ and $D_\text{unsafe}$ constructed in T2I format from BeaverTails~\citep{ji2023beavertails}. We consider two attack settings: (i) harmful queries from StrongREJECT~\citep{souly2024strongreject} converted to T2I format, and (ii) NSFW images from the NSFW Detection dataset~\citep{nsfw_detect_dataset}, all resized to $224\times224$. We focus on NeuroStrike, as it naturally extends to VLMs and is the strongest attack considered. Table~\ref{tab:multimodal} (Appendix) reports ACC on SST2 and ASR under both settings (SR denotes StrongREJECT). 
\alg matches no-defense ACC and outperforms all baselines in ASR, demonstrating strong robustness in multimodal scenarios.

\myparatight{Effectiveness of dynamic neuron identification in \alg}%
A key design choice of \alg is to refresh the safety-critical neuron set periodically rather than fix it once in advance. To isolate its contribution, we replace \alg's dynamic identification with the static identification of SafeNeuron~\citep{wang2026safeneuron}, keeping all other components unchanged. Table~\ref{tab:identification} (Appendix) reports ACC and ASR on SST2 with Llama. The ASR of static variant under NeuroStrike attack rises from 0.04 to 0.29, suggesting that a one-shot estimate becomes stale as safety signals redistribute during fine-tuning.

\myparatight{Sensitivity of \alg to the loss weighting}%
In Eq.~(\ref{eq:total}), $L_{\text{user}}$ and $L_{\text{safe}}$ both carry a coefficient of 1. To check that this choice is not fragile, we introduce a weight $\lambda_{\text{safe}}$ on $L_{\text{safe}}$ purely for this analysis and vary it to $0.5$ and $2$. Table~\ref{tab:lambda_safe} (Appendix) reports ASR and ACC on SST2 with Llama. Results are stable across the range, and \alg keeps ASR at most 0.08 under all six attacks regardless of the weighting. A larger $\lambda_{\text{safe}}$ trades accuracy for marginal ASR gains, so the default $\lambda_{\text{safe}}=1$ already sits at the balance point.


\section{Conclusion} \label{sec:conclusion}

We presented \alg, a fine-tuning-stage defense that redistributes safety signals across a broader set of neurons to address the root cause of alignment fragility. Through ablation-robust optimization, KL-divergence regularization, and randomized gradient projection, \alg hardens LLMs against both jailbreak and neuron-level attacks while preserving task utility, achieving near-zero ASR across diverse attack settings, multimodal scenarios, and white-box adaptive adversaries. These findings highlight the importance of improving safety robustness at the neuron level. We hope \alg can inspire future defenses that make safety alignment more resilient.
\section{Limitations}
\alg operates during the fine-tuning stage and therefore assumes that the defender has full control over the training pipeline, including the training data, model weights, and optimization objectives. This assumption may not hold in deployment scenarios where the model is released without any further fine-tuning, where the defender lacks the computational resources required to perform parameter updates, or where only black-box access to the model is available through an API. In such settings, complementary inference-time defenses would be needed to fill the gap.

\section{Ethical considerations}

The sole purpose of this paper is to advance the safety and robustness of LLMs against post-alignment attacks, with the goal of promoting the responsible deployment of AI systems in real-world applications. We do not encourage, facilitate, or endorse any malicious use of the attack methods discussed in this work. All jailbreak and neuron-level attacks examined in our experiments are drawn from prior publicly available research and are evaluated solely in controlled settings to benchmark the effectiveness of our defense framework, \alg.
In our evaluation pipeline, we employ an LLM to assess attack success rates by determining whether a given model response violates policy. Although our work focuses on defensive alignment, analyzing safety-critical neurons and adaptive attack strategies could potentially inform stronger future attacks. We therefore encourage responsible disclosure and careful deployment of neuron-level safety analyses. The jailbreak prompts and evaluation data used in this work are obtained from publicly available benchmarks and prior research artifacts. We do not collect or use personally identifying information, and all experiments are conducted in controlled research settings involving potentially unsafe or offensive content solely for safety evaluation purposes.

\section*{Acknowledgments}

We thank the reviewers for their constructive comments.
This work was supported by the National Artificial Intelligence Research Resource (NAIRR) Pilot under Award Nos. 250513 and 260142.

\bibliography{refs}

\appendix

\newpage

\begin{algorithm*}[t]
	\caption{Our \alg}
	\label{alg:main}
	\begin{algorithmic}[1]
		\Require User fine-tuning dataset $D_{\text{user}}$, benign query dataset $D_\text{safe}$,
		harmful query dataset $D_\text{unsafe}$, pre-trained model parameters $\theta_0$,
		classifier refresh interval $N$, learning rate $\gamma$
		\Ensure Hardened model $\widehat{\theta}$ with robust safety
		\State Initialize $\theta \leftarrow \theta_0$
		\For{each fine-tuning step $t = 1, 2, \dots$}
		\State \textbf{// Stage 1: Identifying safety-critical neurons}
		\If{$t \bmod N = 1$}
		\State For each layer $l$, train classifier $\hat{y}_l(x)$ from $D_\text{safe} \cup D_\text{unsafe}$ based on Eq.~(\ref{eq:classifier})
		\State Identify $\mathcal{I}_{\text{safety}, l}$ based on Eq.~(\ref{eq:topfrac})
		\EndIf
		\State \textbf{// Stage 2: Ablation-robust safety optimization}
		\State Compute $L_{\text{ref}}$ on harmful batch based on Eq.~(\ref{eq:lr})
		\State Compute $L_{\text{ref}}^{\text{abl}}$ under ablated forward pass based on Eq.~(\ref{eq:lr_abl})
		\State Compute $L_{\text{cons}}$ between ablated and standard passes based on Eq.~(\ref{eq:kl})
		\State Set $L_{\text{safe}} = L_{\text{ref}} + L_{\text{ref}}^{\text{abl}} + L_{\text{cons}}$ based on Eq.~(\ref{eq:defense})
		\State \textbf{// Stage 3: Resolving gradient conflicts via randomized projection}
		\State Compute $g_{\text{user}} = \nabla_{\theta} L_{\text{user}}$, $g_{\text{safe}} = \nabla_{\theta} L_{\text{safe}}$
		\If{$\langle g_{\text{user}},\, g_{\text{safe}} \rangle < 0$}
		\State Sample ordering $(g_1, g_2)$ of $(g_{\text{user}}, g_{\text{safe}})$ and derive $\hat{g}_1$, $\hat{g}_2$ based on Eq.~(\ref{eq:projection})
		\Else
		\State Set $\hat{g}_1 = g_{\text{user}}$, $\hat{g}_2 = g_{\text{safe}}$
		\EndIf
		\State Compute $g_{\text{final}} = \hat{g}_1 + \hat{g}_2$ based on Eq.~(\ref{eq:final_grad})
		\State Update $\theta \leftarrow \theta - \gamma \cdot g_{\text{final}}$
		\EndFor
		\State $\widehat{\theta} \leftarrow \theta$
		\State \Return $\widehat{\theta}$
	\end{algorithmic}
\end{algorithm*}

\begin{figure*}[t]
    \centering
    \includegraphics[width=0.95\textwidth]{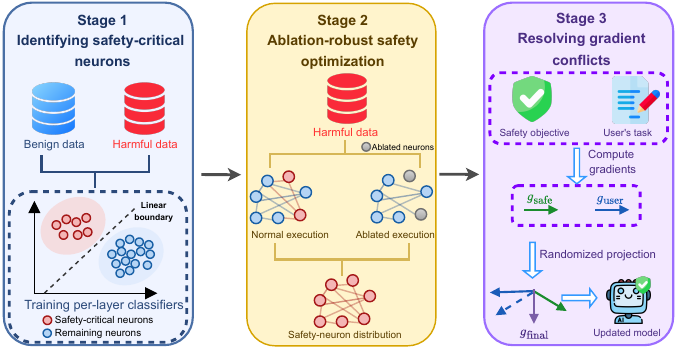}
    \caption{Overview of \alg, which hardens safety alignment against post-alignment attacks via three coupled stages. \textbf{Stage 1} trains per-layer linear classifiers on benign and harmful activations to identify the sparse safety-critical neuron subset. \textbf{Stage 2} deliberately ablates these neurons during training, forcing safety representations to redistribute across a broader neuron population via normal and ablated forward passes with distributional consistency regularization. \textbf{Stage 3} resolves conflicts between the safety gradient $g_{\text{safe}}$ and the user task gradient $g_{\text{user}}$ via randomized projection, producing a final update $g_{\text{final}}$ that preserves both robustness and utility.}
    \label{fig:neuronguard}
\end{figure*}

\FloatBarrier

\section{Discussion of assumptions}
\label{sec:appendix_remarks}

We provide the extended discussion of the assumptions stated in Section~\ref{theoretical_analysis}.

\myparatight{On Assumption~\ref{asmp1} (linear separability and ablation coverage)}
While the exact internal safety boundary and the population-level set $S_l$
are not directly observable, prior work in mechanistic interpretability suggests
that high-level behaviors in neural networks, including safety or
refusal-related behaviors, often concentrate in a sparse subset of specialized
neurons \citep{bau2017network,wu2025neurostrike}.
Separately, empirical studies on activation-based analysis indicate that simple
classifiers trained on intermediate representations can achieve non-trivial
separation between benign and unsafe (or refusal-related) prompts, supporting
approximate linear separability at certain layers \citep{teo2025blessing}.
Assumption~\ref{asmp1} is particularly suitable
for safety-aligned models with stable refusal patterns or narrowly defined unsafe
behaviors, where unsafe prompts induce consistent activation shifts at certain
layers.
In practice, the quantities $\delta_u$, $\delta_b$, and $\rho_l$ can be estimated
on held-out data, allowing the validity of the assumption to be partially verified
empirically rather than taken as an oracle condition.

Note that $\rho_l$ is a theoretical quantity distinct from the hyperparameter
$\rho$: $\rho$ controls the \emph{size} of the selected set, while $\rho_l$
measures the \emph{overlap} with the true safety-neuron set $S_l$.

\myparatight{On Assumption~\ref{asmp4} (bounded KL divergence)}
The direction $D_{\mathrm{KL}}(p_{\theta_{\text{abl}}}\|p_{\theta})$
treats the standard distribution as the reference,
so minimizing it enforces prediction consistency while maintaining consistent
refusal behavior on unsafe prompts.

\myparatight{Notation details}
The safety alignment objective is formally defined as
\begin{align}
\max_{\theta}\ 
\mathbb{E}_{x \sim \mathcal{D}_{\mathrm{unsafe}}}
\ \mathbb{E}_{y \sim p_\theta(\cdot \mid x)}
\big[ R_{\mathrm{refuse}}(x, y) \big],
\end{align}
where $R_{\mathrm{refuse}}$ is a human-defined reward that assigns
higher values to refusal or other policy-compliant responses and
penalizes harmful generations.
The population distributions $\mathcal{D}_{\mathrm{unsafe}}$ and
$\mathcal{D}_{\mathrm{safe}}$ are the theoretical counterparts of
the empirical datasets $D_{\mathrm{unsafe}}$ and $D_{\mathrm{safe}}$
used in Section~\ref{our_method}.

Following \citep{wu2025neurostrike}, we define an internal decision boundary
\begin{align} 
\Upsilon(h_{l}(x);\phi)=
\begin{cases}
1, & x\in\mathcal{D}_{\text{unsafe}},\\
0, & x\in\mathcal{D}_{\text{safe}},
\end{cases} 
\end{align} 
where $\phi\subseteq\theta$ denotes the subset of model parameters responsible
for discriminating between unsafe and benign activations within the layer.
The activation vector at layer $l$ is
$h_{l}(x)=[h_{1,l}(x),\,h_{2,l}(x),\,\ldots,\,h_{d_l,l}(x)]^\top$
with $h_{i,l}(x)$ denoting the activation of neuron $i$ at layer $l$
(following NeuroStrike, we omit the token index for notational simplicity).
$S_l$ is the population-level safety-neuron set, whereas
$\mathcal{I}_{\text{safety},l}$ (defined in Section~\ref{subsec:identify}) is
its empirical approximation obtained via the top-$\rho$ linear-classifier
selection rule.

Given the identified safety-critical neuron set, the ablated forward pass zeros out activations:
\begin{align}
\tilde{h}_{i,l}(x) =
\begin{cases}
0, & \text{if } i \in \mathcal{I}_{\text{safety},l}, \\
h_{i,l}(x), & \text{otherwise}.
\end{cases}
\end{align}
The resulting ablated activation vector $\tilde{h}_{l}(x)$ is then propagated
through the remaining subnetwork from layer $l+1$ to the output layer.

\section{Proof of Theorem~\ref{theorem1}} 
\label{sec:appendix_1}

In this appendix, we provide detailed proofs for the auxiliary lemmas and for the main theorem.
All results are derived under the notation and assumptions in the main text, in particular:
(i) the per-layer safety-neuron set $S_{l}$ defined in Section~\ref{theoretical_analysis}, where $h_{i,l}(x)$ denotes the activation of neuron $i$ at layer $l$;
(ii) the probe-selected safety-critical neuron index set $\mathcal{I}_{\text{safety},l}\subseteq[d_l]$
(identified via the top-$\rho$ selection on classifier weights $W_l$ from Section~\ref{subsec:identify})
with ablation coverage ratio
$
\rho_{l} := |S_{l} \cap \mathcal{I}_{\text{safety},l}|\,/\,|S_{l}|;
$
(iii) the attacker's constrained perturbation set at layer $l$ given in Assumption~\ref{asmp2};
and (iv) all probabilities and expectations are taken over the unsafe input distribution
$\mathcal{D}_{\text{unsafe}}$ unless otherwise stated.

\subsection{Proof of auxiliary lemmas}\label{app:proof_lemma}

\begin{lemma}[Ablation shrinks the adversarial subspace]
Let $S_{l} \subseteq \{1,\dots,d_l\}$ be the true safety-neuron index set at layer $l$,
and let $\mathcal{I}_{\text{safety},l}\subseteq[d_l]$ be the set of safety-critical neuron
indices identified by the probe, where $i\in\mathcal{I}_{\text{safety},l}$ means neuron $i$
is ablated (zeroed out in the forward pass).
Define the ablation coverage ratio as
\begin{align}
    \rho_{l} := \frac{|S_{l} \cap \mathcal{I}_{\text{safety},l}|}{|S_{l}|} \in [0,1].
\end{align}
Under Assumption~\ref{asmp2}, after ablation, the attacker's effective perturbation
set is reduced from $S_{l}$ to
$
S_{l}^{\text{eff}} := S_{l} \setminus \mathcal{I}_{\text{safety},l},
$
and therefore
\begin{align}
    |S_{l}^{\text{eff}}| = |S_{l}|(1 - \rho_{l}).
\end{align}
\label{lemma1}
\end{lemma}

\begin{proof}
For a vector $\delta \in \mathbb{R}^{d_l}$, let $\mathrm{supp}(\delta) := \{\, i \mid \delta_i \neq 0 \,\}$ denote its support. We also denote by $\Delta_l(x) \in \mathbb{R}^{d_l}$ the adversarial perturbation injected at layer $l$ for input $x$, and write $\Delta_{l,i}(x)$ for its $i$-th coordinate. By Assumption~\ref{asmp2}, prior to ablation, the attacker is allowed to perturb only the coordinates in $S_{l}$. Then
\begin{align}
    \nonumber \Delta_{l}(x) \in \mathcal{U}_{l}
:= \{\, \delta \in \mathbb{R}^{d_l} \mid \mathrm{supp}(\delta) \subseteq S_{l},
\\ \ \|\delta\|_{2} \le \varepsilon \,\}.
\end{align}
Therefore, if $i \notin S_{l}$, then necessarily $\Delta_{l,i}(x)=0$.

Our defense produces the ablated activation at each coordinate:
\begin{align}
    \tilde{h}_{i,l}(x) =
\begin{cases}
0, & \text{if } i \in \mathcal{I}_{\text{safety},l},\\
h_{i,l}(x), & \text{otherwise},
\end{cases}
\end{align}
where $h_{i,l}(x)$ is the original activation of neuron $i$ at layer $l$.

For any index with $i \in \mathcal{I}_{\text{safety},l}$, the forward pass always uses $0$
on that coordinate, so any attacker perturbation on this coordinate is nullified.
Hence, for a perturbation component to be effective, the index must (i) belong to
the original attackable set $S_{l}$, and (ii) not be ablated,
i.e.\ $i \notin \mathcal{I}_{\text{safety},l}$.
The set of such coordinates is exactly
$
S_{l}^{\text{eff}} = S_{l} \setminus \mathcal{I}_{\text{safety},l}.
$
By definition of $\rho_{l}$:
\begin{equation}
\begin{aligned}
&|S_{l} \cap \mathcal{I}_{\text{safety},l}| = \rho_{l} |S_{l}|
\\
&\Longrightarrow 
|S_{l}^{\text{eff}}|
= |S_{l}| - \rho_{l} |S_{l}|
= |S_{l}|(1 - \rho_{l}).
\end{aligned}
\end{equation}
Thus, ablation shrinks the adversarial subspace by a factor of $(1-\rho_{l})$.

We now illustrate the probability of successful neuron-level exploitation of our
defense compared to NeuroStrike.

Recall that $\alpha_l \in [0,1]$ denotes the per-neuron attack success probability
(Assumption~\ref{asmp5}), i.e., the probability that an adaptive attacker can
successfully exploit a single available safety neuron in $S_l$.
In the NeuroStrike baseline, no ablation defense is applied, and thus all neurons in $S_l$ remain available for manipulation.
Let $|S_l| = s_l$.
Under the assumption of independent per-neuron attack attempts, the probability
that at least one neuron in $S_l$ is successfully exploited is:
\begin{equation}
P_{nl}
= 1 - (1 - \alpha_l)^{s_l}.
\label{eq:neurostrike-p}
\end{equation}
In our ablation defense, only a fraction $(1-\rho_l)$ of neurons remain non-ablated
and therefore available for attack.
Let the number of such non-ablated neurons be
$s_l^{\mathrm{rem}} = (1-\rho_l)\,s_l$.
The corresponding probability that at least one non-ablated neuron is successfully
exploited is:
\begin{equation}
\begin{aligned}
P_{\mathrm{ours}}
&= 1 - (1 - \alpha_l)^{s_l^{\mathrm{rem}}} \\
&= 1 - (1 - \alpha_l)^{(1-\rho_l)s_l}.
\end{aligned}
\label{eq:ours-p}
\end{equation}
Since $(1-\rho_l)s_l < s_l$ and $\alpha_l \in (0,1)$, we have
\begin{equation}
P_{\mathrm{ours}} < P_{nl}.
\label{eq:prob-comp}
\end{equation}
This inequality rigorously shows that, for the same per-neuron attackability
$\alpha_l$, our ablation defense strictly reduces the probability of successful
neuron-level exploitation compared to NeuroStrike.
\end{proof}

\begin{lemma}[Lipschitz control of output change after ablation/attack]
Let $h_{l}(x)$ be the original activation vector at layer $l$, with components
$h_{i,l}(x)$, and let $\tilde{h}_{l}(x)$ be the ablated activation vector with
components
$\tilde{h}_{i,l}(x) = 0$ if $i \in \mathcal{I}_{\text{safety},l}$, and
$\tilde{h}_{i,l}(x) = h_{i,l}(x)$ otherwise.
Under Assumption~\ref{asmp3}, for any input $x$ and any admissible attack
$\delta$ supported on $S_{l}^{\text{eff}}$ with $\|\delta\|_{2} \le \varepsilon$, we have
\begin{equation}
\begin{aligned}
&\big\| z_{\theta}(\tilde{h}_{l}(x) + \delta) - z_{\theta}(\tilde{h}_{l}(x)) \big\|_{2}
\le C_{L} \, \varepsilon,
\\
&\big\| z_{\theta}({h}_{l}(x) + \delta) - z_{\theta}({h}_{l}(x)) \big\|_{2}
\le C_{L} \, \varepsilon.
\end{aligned}
\end{equation}
\label{lemma2}
\end{lemma}

\begin{proof}
Assumption~\ref{asmp3} states that the subnetwork from layer $l$ to the output is
$C_{L}$-Lipschitz in a neighborhood of the ablated activation:
\begin{align}
    \| z_{\theta}(u) - z_{\theta}(v) \|_{2} \le C_{L} \|u-v\|_{2}
\end{align}
for all $u,v$ in that neighborhood.

Take $u = \tilde{h}_{l}(x) + \delta$ and $v = \tilde{h}_{l}(x)$. Then
\begin{align}
    \| z_{\theta}(u) - z_{\theta}(v) \|_{2}
\le C_{L} \|u - v\|_{2}
& \nonumber \\= C_{L} \|\delta\|_{2}
\le C_{L} \varepsilon,
\end{align}
since $\|\delta\|_{2} \le \varepsilon$ by the attack constraint.
Similarly, we obtain the second result by taking $u = {h}_{l}(x) + \delta$ and
$v = {h}_{l}(x)$.
\end{proof}

\begin{remark}
Lemma~\ref{lemma2} isolates only the effect of the attack on top of the ablated
activation $\tilde{h}_{l}(x)$; the ablation operation itself is absorbed
into $\tilde{h}_{l}(x)$.
\end{remark}

\begin{lemma}[Small expected KL implies small expected $\ell_{1}$ difference]
Under Assumption~\ref{asmp4}, we have
\begin{align}
\mathbb{E}_{x \sim \mathcal{D}_{\text{unsafe}}}
&\big[
\| p_{\theta_{\text{abl}}}(\cdot\mid x)
- p_{\theta}(\cdot\mid x) \|_{1}
\big]
\nonumber \\
&\le \sqrt{2\eta}.
\end{align}
\label{lemma3}
\end{lemma}

\begin{proof}
Let
\begin{align}
Z(x)
&:= D_{\mathrm{KL}}\!\big(
p_{\theta_{\text{abl}}}(y\mid x)
\,\|\,
p_{\theta}(y\mid x)
\big)
\nonumber \\
&\ge 0.
\end{align}
By Assumption~\ref{asmp4} we know that
$
\mathbb{E}_{x \sim \mathcal{D}_{\text{unsafe}}}[Z(x)] \;\le\; \eta.
$

We now relate the KL divergence $Z(x)$ to the $\ell_{1}$ (total variation) distance
of the two output distributions.
For any fixed $x$, since $p_\theta(\cdot\mid x)$ and $p_{\theta_{\text{abl}}}(\cdot\mid x)$
are probability distributions over the same measurable output space, Pinsker's
inequality states that for two probability distributions $P$ and $Q$,
\begin{align}
    \|P - Q\|_{1} \;\le\; \sqrt{ 2\, D_{\mathrm{KL}}(P \| Q) }.
\end{align}
Applying this to $P = p_{\theta_{\text{abl}}}(\cdot\mid x)$ and
$Q = p_{\theta}(\cdot\mid x)$, we obtain for every $x$:
\begin{align}
    \big\| p_{\theta_{\text{abl}}}(\cdot\mid x) - p_{\theta}(\cdot\mid x) \big\|_{1}
\;\le\; \sqrt{ 2\, Z(x) }.
\end{align}
Taking expectation over $x \sim \mathcal{D}_{\text{unsafe}}$ on both sides yields
\begin{align}
\mathbb{E}_{x}\!
\big[
\| p_{\theta_{\text{abl}}}(\cdot\mid x)
- p_{\theta}(\cdot\mid x) \|_{1}
\big]
\nonumber \\
\le\;
\mathbb{E}_{x}\!
\big[
\sqrt{2\,Z(x)}
\big].
\end{align}
We upper-bound the right-hand side using Jensen's inequality.
Since $f(t) = \sqrt{t}$ is concave on $t \ge 0$, for any nonnegative random
variable $W$ we have $\mathbb{E}[ \sqrt{W} ] \le \sqrt{ \mathbb{E}[W] }$.
Letting $W = 2 Z(x)$:
\begin{align}
\mathbb{E}_{x}
\big[
\sqrt{2\,Z(x)}
\big]
&\le
\sqrt{
\mathbb{E}_{x}
\big[
2\,Z(x)
\big]
}
\nonumber \\
&=
\sqrt{
2\,\mathbb{E}_{x}
\big[
Z(x)
\big]
}.
\end{align}
Combining and invoking Assumption~\ref{asmp4} ($\mathbb{E}_{x}[Z(x)] \le \eta$):
\begin{align}
    \mathbb{E}_{x}
\big[ \| p_{\theta_{\text{abl}}}(\cdot\mid x) - p_{\theta}(\cdot\mid x) \|_{1} \big]
\;\le\;
\sqrt{ 2\, \eta }.
\end{align}
This proves the lemma.
\end{proof}

\subsection{Proof for Theorem~\ref{theorem1}}\label{app:proof_theorem}

\begin{proof}
We first recall the two ASR upper bounds derived under
Assumptions~\ref{asmp1}--\ref{asmp5}.

\myparatight{ASR of the NeuroStrike baseline}
By Lemma~\ref{lemma1}, before our ablation defense, the attacker can exploit any of
the $s_l = |S_l|$ safety neurons, and under the per-neuron success model with
probability $\alpha_l$ this happens with probability
\begin{align}
    P_{nl} = 1 - (1 - \alpha_l)^{s_l}.
\end{align}
By the Lipschitz control before ablation (Lemma~\ref{lemma2} applied to the standard
activation $h_l(x)$), an unsafe output can occur either due to successful
neuron-level exploitation or due to the insufficient clean safety margin under a
Lipschitz-bounded perturbation.
By a union bound over these two failure modes, we obtain the upper bound:
\begin{equation}
\begin{aligned}
\mathcal{B}(\mathrm{ASR}_{nl})
&=
P_{nl}
\\
&\quad+
\Pr_{x \sim \mathcal{D}_{\mathrm{unsafe}}}
\big[
m(x) \le C_L \varepsilon
\big].
\end{aligned}
\label{eq:proof-asr-ns}
\end{equation}
where $m(x)$ denotes the clean safety margin of the standard model on input $x$.

\myparatight{ASR of our defense}
With our ablation, Lemma~\ref{lemma1} shows the effective attackable neurons are
reduced to $(1-\rho_l)s_l$, hence the corresponding neuron-level success
probability is
\begin{align}
    P_{\mathrm{ours}} = 1 - (1 - \alpha_l)^{(1-\rho_l)s_l}.
\end{align}
An attack can still lead to unsafe generations after ablation.
Therefore, we upper bound the attack success rate by the probability of sampling an unsafe output under the ablated model.
Taking expectation over $x\sim\mathcal{D}_{\mathrm{unsafe}}$, and we denote by $\mathcal{Y}_{\mathrm{unsafe}}$ the set of all model outputs deemed unsafe or policy-violating under the safety criterion. We obtain the following upper bound on the expected ASR:
\begin{equation}
\begin{aligned}
\mathcal{B}(\mathrm{ASR}_{\mathrm{ours}})
&=
P_{\mathrm{ours}}
\\
&\quad+
\mathbb{E}_{x}
\Pr_{y\sim p_{\theta_{\text{abl}}}(\cdot\mid x)}
\big[
y\in\mathcal{Y}_{\mathrm{unsafe}}
\big].
\end{aligned}
\label{eq:proof-asr-ours}
\end{equation}
Moreover, for any unsafe prompt $x$, by the definition of total variation distance:
\begin{equation}
\begin{aligned}
&
\Pr_{\substack{
y\sim p_{\theta_{\text{abl}}}(\cdot\mid x)
}}
\big[
y\in\mathcal{Y}_{\mathrm{unsafe}}
\big]
\\
&\le
\Pr_{\substack{
y\sim p_\theta(\cdot\mid x)
}}
\big[
y\in\mathcal{Y}_{\mathrm{unsafe}}
\big]
\\
&\quad+
\big\|
p_{\theta_{\text{abl}}}(\cdot\mid x)
-
p_\theta(\cdot\mid x)
\big\|_1 .
\end{aligned}
\label{eq:prunsafe}
\end{equation}
Taking expectation over $x\sim\mathcal{D}_{\mathrm{unsafe}}$ and applying
Lemma~\ref{lemma3} yields:
\begin{align}
\mathbb{E}_{x}\Pr_{y\sim p_{\theta_{\text{abl}}}(\cdot\mid x)}[y\in\mathcal{Y}_{\mathrm{unsafe}}]
& \le
\Pr[m(x)\le C_L\varepsilon] \nonumber \\ 
& +
\sqrt{2\eta}.
\label{eq:expectation}
\end{align}

\myparatight{Taking the difference}
Subtracting the two ASR upper bounds yields:
\begin{equation}
\begin{aligned}
&\mathcal{B}(\mathrm{ASR}_{nl}) - \mathcal{B}(\mathrm{ASR}_{\mathrm{ours}})
\\
&=
\big( P_{nl} - P_{\mathrm{ours}} \big)
+
\Pr[m(x) \le C_L \varepsilon]
\\
&-
\mathbb{E}_{x}\Pr_{y\sim p_{\theta_{\text{abl}}}(\cdot\mid x)}[y\in\mathcal{Y}_{\mathrm{unsafe}}].
\label{eq:diff}
\end{aligned}
\end{equation}
Applying \eqref{eq:expectation} to \eqref{eq:diff} yields:
\begin{align}
\mathcal{B}(\mathrm{ASR}_{nl})
-
\mathcal{B}(\mathrm{ASR}_{\mathrm{ours}})
&\ge
\big(
P_{nl}
-
P_{\mathrm{ours}}
\big)
\nonumber \\
&\quad
-
\sqrt{2\eta}.
\end{align}
Recalling the neuron-level exploitation probabilities in Lemma~\ref{lemma1},
where $s_l = |S_l|$ and $\rho_l$ is the ablation coverage ratio:
\begin{equation}
\begin{aligned}
P_{nl} - P_{\mathrm{ours}}
&=
\Big(1 - (1-\alpha_l)^{s_l}\Big)
\\
&\quad
-
\Big(
1
-
(1-\alpha_l)^{(1-\rho_l)s_l}
\Big)
\\
&=
(1-\alpha_l)^{(1-\rho_l)s_l}
-
(1-\alpha_l)^{s_l}.
\end{aligned}
\label{eq:p-gap}
\end{equation}
Since $\rho_l\in(0,1)$ and $\alpha_l\in(0,1)$, we have $(1-\rho_l)s_l < s_l$ and
the map $t\mapsto (1-\alpha_l)^t$ is strictly decreasing in $t$, which implies
\begin{align}
(1-\alpha_l)^{(1-\rho_l)s_l} - (1-\alpha_l)^{s_l} > 0.
\end{align}
Combining this with the previous inequality yields:
\begin{equation}
\begin{aligned}
&\mathcal{B}(\mathrm{ASR}_{nl}) - \mathcal{B}(\mathrm{ASR}_{\mathrm{ours}})
\\
&\ge
\Big((1-\alpha_l)^{(1-\rho_l)s_l} - (1-\alpha_l)^{s_l}\Big)
-
\sqrt{2\eta}.
\label{eq:final-gap}
\end{aligned}
\end{equation}
Hence, under the sufficient condition
\begin{align}
\sqrt{2\eta}\ \le\ (1-\alpha_l)^{(1-\rho_l)s_l} - (1-\alpha_l)^{s_l},
\label{eq:suff-cond}
\end{align}
we obtain
\begin{align}
\mathcal{B}(\mathrm{ASR}_{nl}) - \mathcal{B}(\mathrm{ASR}_{\mathrm{ours}})\ge 0,
\end{align}
i.e., $\mathcal{B}(\mathrm{ASR}_{nl}) \ge \mathcal{B}(\mathrm{ASR}_{\mathrm{ours}})$.
This completes the proof.
\end{proof}

\section{License of models and datasets}
\label{app:licenses}
We summarize the licenses of all models and datasets used in this paper.

\myparatight{Models}%
Llama-3.1-8B-Instruct~\citep{grattafiori2024llama} and Llama-Guard-3-8B~\citep{grattafiori2024llama} are both released under the Llama 3.1 Community License, a custom commercial license provided by Meta Platforms, Inc.
Qwen2.5-7B-Instruct~\citep{qwen2.5} and Qwen2.5-VL-7B-Instruct~\citep{qwen2.5-VL} are released under the Apache License 2.0.
Falcon3-7B-Instruct~\citep{Falcon3} is released under the TII Falcon License 2.0.

\myparatight{Datasets}%
SST2~\citep{socher2013recursive} is made available by Stanford for research purposes; 
no explicit license is stated on the official website.
AGNews~\citep{zhang2015character} is provided by the academic community for research purposes, including data mining, information retrieval, and other non-commercial activities, with no explicit open-source license attached.
CoLA~\citep{warstadt2019neural} consists of sentences excerpted from published linguistics literature; the corpus and baseline code are made available under the MIT License.
GSM8K~\citep{cobbe2021GSM8K} is released by OpenAI under the MIT License.
StrongREJECT~\citep{souly2024strongreject} releases its custom-generated data and code under the MIT License.
The NSFW Detection dataset~\citep{nsfw_detect_dataset} is released under the MIT License.

\section{Details of baselines}
\label{sec:appendix_baseline}

\myparatight{Perplexity~\citep{alon2023detecting}}%
Perplexity serves as the filter by computing the model's own perplexity on the prompt and flags inputs whose perplexity exceeds a threshold to determine whether the prompt is a jailbreak attempt.

\myparatight{SmoothLLM~\citep{robey2023smoothllm}}%
SmoothLLM is a character-level defense that mitigates jailbreaking attacks by exploiting the fragility of adversarial prompts to random perturbations. It creates $N$ perturbed copies of an input prompt and passes them through the target LLM, then aggregates the resulting responses through a majority vote to detect and filter out adversarial inputs without requiring model retraining.

\myparatight{GradSafe~\citep{xie2024gradsafe}}%
GradSafe identifies jailbreak prompts by analyzing the gradients of an LLM's safety-critical parameters when a prompt is paired with a compliance response. The method utilizes the observation that adversarial prompts exhibit consistent gradient patterns across these parameters, which differ significantly from those of safe prompts.

\myparatight{CAT~\citep{xhonneux2024efficient}}%
CAT utilizes adversarial training in the continuous embedding space and uses a lightweight surrogate model to select the most effective adversarial perturbations from $K$ candidates before mapping them back to the discrete token space for fine-tuning, achieving robustness against various jailbreak attacks while maintaining the model's nominal performance.

\myparatight{LED~\citep{zhao2024defending}}%
LED is an editing-based defense that mitigates jailbreak attacks by identifying and editing specific ``safety-critical'' layers within an LLM. The method works by analyzing the differences in hidden state activations between safe and adversarial prompts to locate the layers most responsible for safety alignment.

\myparatight{SafeNeuron~\citep{wang2026safeneuron}}%
SafeNeuron identifies safety-critical neurons by comparing activations under safe and unsafe inputs using Activation Effect Size and Safety Activation Shift metrics, then freezes these neurons during direct preference optimization to force the model to construct redundant safety pathways across the remaining parameters, improving robustness against neuron-level pruning attacks.

\section{Details of adaptive attacks} 
\label{app:Adaptive_Attacks}

\myparatight{Iterative pruning (IP)}%
This attack directly targets the identification step described in Section~\ref{subsec:identify}. Specifically, the attacker trains per-layer linear classifiers in the same manner as \alg to identify the safety-neuron index sets $\mathcal{I}_{\text{safety}, l}$. The attack then proceeds in iterative rounds: in each round, the attacker re-trains the classifiers on the ablated model to identify a new set of neurons that have assumed safety-related responsibilities, and then permanently removes them by zeroing out the corresponding positions. We set the number of pruning rounds to 5.

\myparatight{Nonlinear evasion (NE)}%
This attack exploits the limitation of our neuron identification, which relies on linear classifiers. Specifically, instead of training per-layer linear classifiers, the attacker trains a two-layer MLP for each layer to search for alternative safety-neuron sets that capture nonlinear neuron interactions. The attack then ablates the identified neurons at inference time, targeting safety-relevant neurons that lie outside the linearly identified set and may not be captured by \alg.

\section{Robustness under adaptive attacks} 
\label{app:Adaptive_Attacks_analy}

\alg's resilience against adaptive adversaries is rooted in the fundamental asymmetry between what the attacker must accomplish and what the defense requires to succeed. Because safety signals are continuously redistributed across an ever-wider population of neurons throughout fine-tuning, an adversary with full knowledge of the defense architecture cannot simply target a fixed, locatable set of neurons and expect lasting success. By the time an attacker has identified and suppressed one generation of safety-bearing neurons, the training process has already propagated those representations into new neurons elsewhere in the network, effectively moving the goalposts with each successive step.

The iterative pruning attack illustrates this asymmetry clearly. This adaptive strategy mirrors \alg's own identification logic, using linear classifiers to locate safety neurons and removing them in successive rounds on the progressively weakened model. Yet the attacker must achieve complete elimination of safety behavior across all redundant sites, while the defense only needs to preserve it in a sufficient fraction of them. The ablation-robust optimization that \alg undergoes during fine-tuning is structurally identical to the pressure this attack applies, meaning the model has been explicitly conditioned to withstand repeated targeted suppression of its safety neurons. The result is an ASR of just 0.09, compared to 0.83 without any defense.

The nonlinear evasion attack takes a different approach by training MLPs to uncover safety-relevant neurons that interact in ways beyond what linear classifiers can detect, attempting to exploit the boundary of \alg's identification capability. Even so, the defense holds because its robustness does not hinge on neuron identification being exhaustive. Partial ablation coverage alone is sufficient to meaningfully shrink the set of neurons available for adversarial exploitation, and the KL-divergence regularization further constrains how far the ablated model's behavior can drift from the standard model. These two properties together provide a robustness floor that persists even when the attacker successfully probes beyond the linearly identifiable safety neurons, keeping ASR at just 0.05.

\begin{table*}[t]
\centering
\small
\addtolength{\tabcolsep}{-1pt}
\begin{tabular}{l l | cccccccc}
\toprule
Model & Attack & No defense & Perplexity & SmoothLLM & GradSafe & CAT & LED & SafeNeuron & \cellcolor{greyL}\alg \\
\midrule
\multirow{6}{*}{Llama}
& PAIR & 0.15 & 0.14 & 0.15 & 0.09 & 0.10 & 0.07 & 0.04 & \cellcolor{greyL}0.01 \\
& TAP & 0.15 & 0.13 & 0.11 & 0.08 & 0.12 & 0.09 & 0.05 & \cellcolor{greyL}0.00 \\
& Puzzler & 0.28 & 0.25 & 0.23 & 0.12 & 0.17 & 0.13 & 0.07 & \cellcolor{greyL}0.00 \\
& GCG & 0.35 & 0.16 & 0.27 & 0.15 & 0.09 & 0.05 & 0.04 & \cellcolor{greyL}0.01 \\
& AD & 0.31 & 0.28 & 0.26 & 0.11 & 0.13 & 0.09 & 0.05 & \cellcolor{greyL}0.01 \\
& NS & 0.89 & 0.65 & 0.86 & 0.72 & 0.58 & 0.49 & 0.22 & \cellcolor{greyL}0.02 \\
\midrule
\multirow{6}{*}{Qwen}
& PAIR & 0.23 & 0.18 & 0.19 & 0.14 & 0.18 & 0.17 & 0.06 & \cellcolor{greyL}0.00 \\
& TAP & 0.25 & 0.21 & 0.10 & 0.16 & 0.17 & 0.13 & 0.07 & \cellcolor{greyL}0.03 \\
& Puzzler & 0.41 & 0.37 & 0.39 & 0.25 & 0.29 & 0.20 & 0.10 & \cellcolor{greyL}0.01 \\
& GCG & 0.40 & 0.17 & 0.33 & 0.17 & 0.13 & 0.08 & 0.04 & \cellcolor{greyL}0.01 \\
& AD & 0.36 & 0.35 & 0.23 & 0.18 & 0.20 & 0.15 & 0.09 & \cellcolor{greyL}0.01 \\
& NS & 0.83 & 0.57 & 0.70 & 0.68 & 0.53 & 0.42 & 0.21 & \cellcolor{greyL}0.00 \\
\midrule
\multirow{6}{*}{Falcon}
& PAIR & 0.15 & 0.12 & 0.12 & 0.08 & 0.05 & 0.02 & 0.03 & \cellcolor{greyL}0.01 \\
& TAP & 0.16 & 0.15 & 0.15 & 0.06 & 0.07 & 0.03 & 0.02 & \cellcolor{greyL}0.00 \\
& Puzzler & 0.30 & 0.25 & 0.28 & 0.09 & 0.09 & 0.05 & 0.04 & \cellcolor{greyL}0.01 \\
& GCG & 0.28 & 0.13 & 0.24 & 0.13 & 0.16 & 0.12 & 0.06 & \cellcolor{greyL}0.00 \\
& AD & 0.22 & 0.20 & 0.14 & 0.11 & 0.13 & 0.09 & 0.04 & \cellcolor{greyL}0.01 \\
& NS & 0.81 & 0.70 & 0.74 & 0.59 & 0.66 & 0.57 & 0.28 & \cellcolor{greyL}0.02 \\
\bottomrule
\end{tabular}
\caption{ASR of different methods under six attack types across three models on the AGNews task.}
\label{tab:jailbreak_asr_agnews}

\end{table*}

\begin{table*}[t]
\centering
\small
\addtolength{\tabcolsep}{-1pt}
\begin{tabular}{l l | cccccccc}
\toprule
Model & Attack & No defense & Perplexity & SmoothLLM & GradSafe & CAT & LED & SafeNeuron & \cellcolor{greyL}\alg \\
\midrule
\multirow{6}{*}{Llama}
& PAIR & 0.17 & 0.16 & 0.15 & 0.07 & 0.10 & 0.07 & 0.04 & \cellcolor{greyL}0.01 \\
& TAP & 0.15 & 0.13 & 0.13 & 0.08 & 0.15 & 0.09 & 0.05 & \cellcolor{greyL}0.00 \\
& Puzzler & 0.29 & 0.25 & 0.23 & 0.12 & 0.19 & 0.13 & 0.07 & \cellcolor{greyL}0.01 \\
& GCG & 0.35 & 0.16 & 0.27 & 0.15 & 0.06 & 0.05 & 0.03 & \cellcolor{greyL}0.01 \\
& AD & 0.31 & 0.28 & 0.24 & 0.11 & 0.10 & 0.09 & 0.05 & \cellcolor{greyL}0.01 \\
& NS & 0.89 & 0.64 & 0.83 & 0.72 & 0.58 & 0.46 & 0.22 & \cellcolor{greyL}0.04 \\
\midrule
\multirow{6}{*}{Qwen}
& PAIR & 0.26 & 0.18 & 0.19 & 0.14 & 0.18 & 0.14 & 0.06 & \cellcolor{greyL}0.01 \\
& TAP & 0.24 & 0.21 & 0.10 & 0.16 & 0.17 & 0.13 & 0.07 & \cellcolor{greyL}0.01 \\
& Puzzler & 0.41 & 0.37 & 0.39 & 0.25 & 0.26 & 0.20 & 0.10 & \cellcolor{greyL}0.00 \\
& GCG & 0.40 & 0.18 & 0.33 & 0.17 & 0.13 & 0.08 & 0.04 & \cellcolor{greyL}0.00 \\
& AD & 0.36 & 0.38 & 0.23 & 0.18 & 0.20 & 0.15 & 0.04 & \cellcolor{greyL}0.01 \\
& NS & 0.83 & 0.57 & 0.68 & 0.68 & 0.54 & 0.42 & 0.21 & \cellcolor{greyL}0.00 \\
\midrule
\multirow{6}{*}{Falcon}
& PAIR & 0.15 & 0.12 & 0.12 & 0.05 & 0.05 & 0.02 & 0.03 & \cellcolor{greyL}0.01 \\
& TAP & 0.15 & 0.15 & 0.15 & 0.08 & 0.07 & 0.04 & 0.02 & \cellcolor{greyL}0.00 \\
& Puzzler & 0.30 & 0.26 & 0.27 & 0.09 & 0.08 & 0.05 & 0.03 & \cellcolor{greyL}0.01 \\
& GCG & 0.28 & 0.11 & 0.24 & 0.15 & 0.16 & 0.12 & 0.05 & \cellcolor{greyL}0.01 \\
& AD & 0.22 & 0.20 & 0.16 & 0.09 & 0.13 & 0.07 & 0.05 & \cellcolor{greyL}0.01 \\
& NS & 0.81 & 0.71 & 0.74 & 0.59 & 0.66 & 0.57 & 0.26 & \cellcolor{greyL}0.02 \\
\bottomrule
\end{tabular}
\caption{ASR of different methods under six attack types across three models on the CoLA task.}
\label{tab:jailbreak_asr_cola}

\end{table*}

\begin{table*}[t]
\centering
\small
\addtolength{\tabcolsep}{-1pt}
\begin{tabular}{l l | cccccccc}
\toprule
Model & Attack & No defense & Perplexity & SmoothLLM & GradSafe & CAT & LED & SafeNeuron & \cellcolor{greyL}\alg \\
\midrule
\multirow{6}{*}{Llama}
& PAIR & 0.17 & 0.14 & 0.15 & 0.10 & 0.10 & 0.07 & 0.04 & \cellcolor{greyL}0.01 \\
& TAP & 0.15 & 0.13 & 0.13 & 0.08 & 0.12 & 0.09 & 0.06 & \cellcolor{greyL}0.00 \\
& Puzzler & 0.29 & 0.23 & 0.23 & 0.12 & 0.17 & 0.13 & 0.07 & \cellcolor{greyL}0.01 \\
& GCG & 0.32 & 0.16 & 0.27 & 0.15 & 0.09 & 0.05 & 0.03 & \cellcolor{greyL}0.01 \\
& AD & 0.31 & 0.28 & 0.24 & 0.11 & 0.13 & 0.09 & 0.03 & \cellcolor{greyL}0.01 \\
& NS & 0.91 & 0.66 & 0.83 & 0.72 & 0.58 & 0.49 & 0.22 & \cellcolor{greyL}0.04 \\
\midrule
\multirow{6}{*}{Qwen}
& PAIR & 0.23 & 0.18 & 0.19 & 0.14 & 0.18 & 0.15 & 0.09 & \cellcolor{greyL}0.00 \\
& TAP & 0.23 & 0.21 & 0.10 & 0.16 & 0.17 & 0.13 & 0.07 & \cellcolor{greyL}0.01 \\
& Puzzler & 0.39 & 0.37 & 0.39 & 0.25 & 0.26 & 0.20 & 0.10 & \cellcolor{greyL}0.00 \\
& GCG & 0.40 & 0.16 & 0.33 & 0.17 & 0.13 & 0.08 & 0.04 & \cellcolor{greyL}0.01 \\
& AD & 0.36 & 0.35 & 0.23 & 0.18 & 0.22 & 0.15 & 0.07 & \cellcolor{greyL}0.01 \\
& NS & 0.83 & 0.58 & 0.68 & 0.68 & 0.55 & 0.42 & 0.21 & \cellcolor{greyL}0.02 \\
\midrule
\multirow{6}{*}{Falcon}
& PAIR & 0.15 & 0.12 & 0.12 & 0.07 & 0.03 & 0.02 & 0.03 & \cellcolor{greyL}0.01 \\
& TAP & 0.16 & 0.15 & 0.15 & 0.08 & 0.07 & 0.04 & 0.02 & \cellcolor{greyL}0.00 \\
& Puzzler & 0.30 & 0.25 & 0.27 & 0.07 & 0.09 & 0.05 & 0.03 & \cellcolor{greyL}0.01 \\
& GCG & 0.28 & 0.09 & 0.24 & 0.13 & 0.16 & 0.12 & 0.06 & \cellcolor{greyL}0.03 \\
& AD & 0.22 & 0.20 & 0.16 & 0.11 & 0.10 & 0.09 & 0.05 & \cellcolor{greyL}0.01 \\
& NS & 0.81 & 0.71 & 0.74 & 0.59 & 0.65 & 0.57 & 0.28 & \cellcolor{greyL}0.02 \\
\bottomrule
\end{tabular}
\caption{ASR of different methods under six attack types across three models on the GSM8K task.}
\label{tab:jailbreak_asr_gsm8k}

\end{table*}

\begin{figure*}[t]
    \centering
    \includegraphics[width=0.8\linewidth]{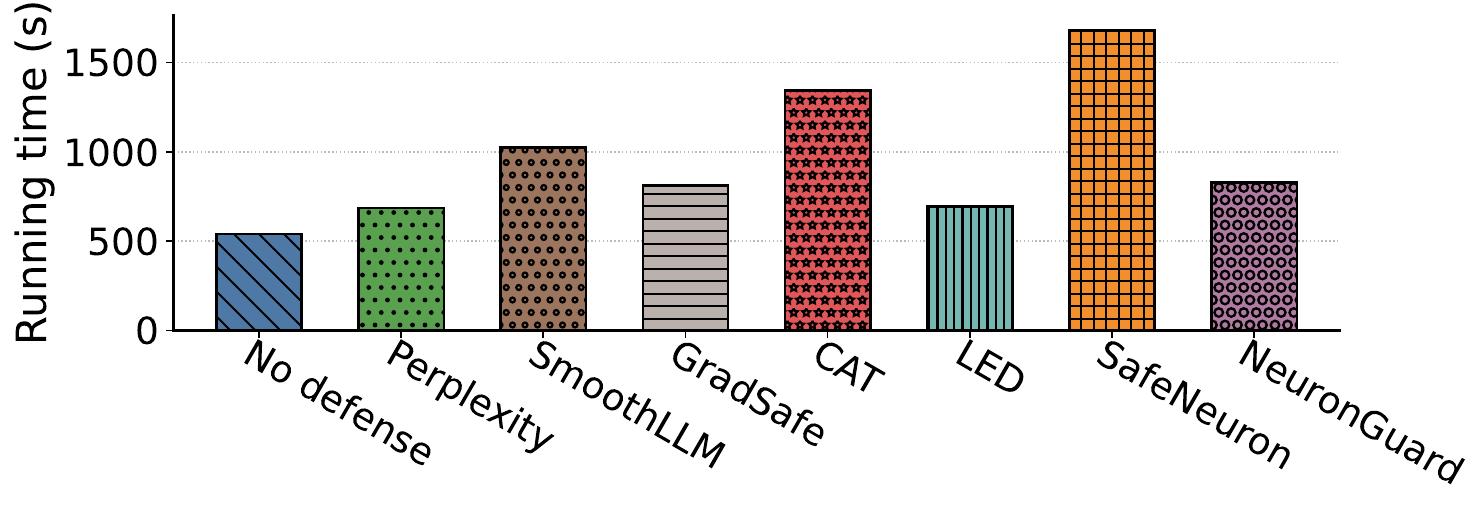}
    \caption{Overhead (seconds) for different methods.}
    \label{fig:computational_overhead}
\end{figure*}

\begin{table*}[t]
\centering
\small
\begin{tabular}{c | cccccc | c}
\toprule
\multirow{2}{*}{$\rho$} 
& PAIR & TAP & Puzzler & GCG & AutoDAN & NeuroStrike & SST2 \\
\cmidrule(lr){2-8}
& ASR & ASR & ASR & ASR & ASR & ASR & ACC \\
\midrule
0.01
& 0.08 & 0.07 & 0.09 & 0.08 & 0.07 & 0.10  & 0.93  \\
0.05
& 0.01 & 0.00 & 0.01 & 0.01 & 0.01 & 0.04  & 0.92  \\
0.10
& 0.01 & 0.00 & 0.01 & 0.01 & 0.01 & 0.01  & 0.92  \\
0.15 
& 0.00 & 0.00 & 0.01 & 0.00 & 0.01 & 0.01  & 0.91  \\
0.20  
& 0.00 & 0.00 & 0.00 & 0.00 & 0.00 & 0.00  & 0.88  \\
0.25 
& 0.00 & 0.00 & 0.00 & 0.00 & 0.00 & 0.00  & 0.84  \\
\bottomrule
\end{tabular}
\caption{Results for the fraction of masked neurons $\rho$ under different attacks on the SST2 task with the Llama model.}
\label{tab:topk}
\end{table*}

\begin{table*}[t]
\centering
\small
\begin{tabular}{c | cccccc | c c}
\toprule
\multirow{2}{*}{$N$} 
& PAIR & TAP & Puzzler & GCG & AutoDAN & NeuroStrike & SST2 & Time \\
\cmidrule(lr){2-9}
& ASR & ASR & ASR & ASR & ASR & ASR & ACC & seconds \\
\midrule
50
& 0.01 & 0.00 & 0.01 & 0.01 & 0.01 & 0.01  & 0.88 & 2628.41 \\
100 
& 0.01 & 0.00 & 0.01 & 0.01 & 0.01 & 0.01  & 0.90 & 1426.58 \\
200 
& 0.01 & 0.00 & 0.01 & 0.01 & 0.01 & 0.04  & 0.93 & 829.18 \\
400  
& 0.04 & 0.03 & 0.05 & 0.04 & 0.04 & 0.10  & 0.93 & 673.45 \\
800 
& 0.14 & 0.12 & 0.17 & 0.15 & 0.13 & 0.26  & 0.93 & 619.49 \\
\bottomrule
\end{tabular}
\caption{Results for the refresh interval $N$ under different attacks on the SST2 task with the Llama model.}
\label{tab:refresh_N}
\end{table*}

\begin{table*}[t]
\centering
\small
\begin{tabular}{c | cccccc | c}
\toprule
\multirow{2}{*}{Size} 
& PAIR & TAP & Puzzler & GCG & AutoDAN & NeuroStrike & SST2 \\
\cmidrule(lr){2-8}
& ASR & ASR & ASR & ASR & ASR & ASR & ACC \\
\midrule
500
& 0.06 & 0.05 & 0.07 & 0.06 & 0.06 & 0.09  & 0.92  \\
750
& 0.01 & 0.00 & 0.01 & 0.01 & 0.01 & 0.04  & 0.92  \\
1250
& 0.01 & 0.00 & 0.01 & 0.01 & 0.01 & 0.05  & 0.90  \\
2500 
& 0.01 & 0.00 & 0.01 & 0.01 & 0.00 & 0.01  & 0.89  \\
\bottomrule
\end{tabular}
\caption{Results of \alg with varying sizes of $D_{\text{safe}}$ and $D_{\text{unsafe}}$ during fine-tuning under different attacks on SST2 with the Llama model.}
\label{tab:safe_validation}
\end{table*}

\begin{table*}[t]
\centering
\small
\begin{tabular}{l | cccccc | c}
\toprule
\multirow{2}{*}{Variant} 
& PAIR & TAP & Puzzler & GCG & AutoDAN & NeuroStrike & SST2 \\
\cmidrule(lr){2-8}
& ASR & ASR & ASR & ASR & ASR & ASR & ACC \\
\midrule
Variant I
& 0.08 & 0.07 & 0.11 & 0.08 & 0.09 & 0.71  & 0.92  \\
Variant II
& 0.10 & 0.09 & 0.13 & 0.11 & 0.10 & 0.78  & 0.89  \\
Variant III
& 0.11 & 0.10 & 0.15 & 0.13 & 0.12 & 0.80  & 0.59  \\
Variant IV
& 0.08 & 0.06 & 0.14 & 0.16 & 0.17 & 0.22  & 0.66  \\

\rowcolor{greyL}
\alg
& 0.01 & 0.00 & 0.01 & 0.01 & 0.01 & 0.04  & 0.92  \\
\bottomrule
\end{tabular}
\caption{Results of different variants of \alg under different attacks on the SST2 task with the Llama model.}
\label{tab:variants}
\end{table*}

\begin{table*}[t]
\centering
\small
\begin{tabular}{l |cccccccc}
\toprule
Attack & No defense & Perplexity & SmoothLLM & GradSafe & CAT & LED & SafeNeuron & \cellcolor{greyL}\alg \\
\midrule
IP & 0.83 & 0.59 & 0.64 & 0.62 & 0.57 & 0.39 & 0.42 & \cellcolor{greyL}0.09 \\
NE  & 0.75 & 0.52 & 0.45 & 0.49 & 0.45 & 0.43 & 0.33 & \cellcolor{greyL}0.05 \\
\bottomrule
\end{tabular}
\caption{ASR of different methods under adaptive attacks on the SST2 task with the Llama model.}
\label{tab:adaptive_attack}
\end{table*}

\begin{table*}[t]
\centering
\small
\addtolength{\tabcolsep}{-1pt}
\begin{tabular}{l l |cccccccc}
\toprule
Dataset & Metric
& No defense 
& Perplexity 
& SmoothLLM 
& GradSafe 
& CAT 
& LED 
& SafeNeuron
& \cellcolor{greyL}\alg \\
\midrule
SST2 & ACC & 0.93 & 0.92 & 0.75 & 0.90 & 0.85& 0.85& 0.86& \cellcolor{greyL}0.89\\
\midrule
SR   & ASR & 0.99 & 0.79 & 0.75 & 0.70 & 0.68& 0.53 & 0.37& \cellcolor{greyL}0.02 \\
NSFW & ASR & 0.96 & 0.82 & 0.74 & 0.67 & 0.60& 0.75 & 0.44& \cellcolor{greyL}0.01 \\
\bottomrule
\end{tabular}
\caption{Results of different methods in multimodal scenarios under the NeuroStrike attack on the SST2 task.}
\label{tab:multimodal}
\end{table*}

\begin{table*}[t]
\centering
\small
\begin{tabular}{l | cccccc | c}
\toprule
\multirow{2}{*}{Identification strategy} 
& PAIR & TAP & Puzzler & GCG & AutoDAN & NeuroStrike & SST2 \\
\cmidrule(lr){2-8}
& ASR & ASR & ASR & ASR & ASR & ASR & ACC \\
\midrule
SafeNeuron
& 0.04 & 0.05 & 0.07 & 0.03 & 0.05 & 0.22  & 0.87  \\
\alg w/ static identification
& 0.04 & 0.03 & 0.05 & 0.04 & 0.05 & 0.29  & 0.92  \\

\rowcolor{greyL}
\alg
& 0.01 & 0.00 & 0.01 & 0.01 & 0.01 & 0.04  & 0.92  \\
\bottomrule
\end{tabular}
\caption{Results of static versus dynamic identification under different attacks on the SST2 task with the Llama model.}
\label{tab:identification}
\end{table*}

\begin{table*}[t]
\centering
\small
\begin{tabular}{c | cccccc | c}
\toprule
\multirow{2}{*}{$\lambda_{\text{safe}}$} 
& PAIR & TAP & Puzzler & GCG & AutoDAN & NeuroStrike & SST2 \\
\cmidrule(lr){2-8}
& ASR & ASR & ASR & ASR & ASR & ASR & ACC \\
\midrule
0.5
& 0.03 & 0.01 & 0.03 & 0.02 & 0.02 & 0.08  & 0.93  \\
1
& 0.01 & 0.00 & 0.01 & 0.01 & 0.01 & 0.04  & 0.92  \\
2
& 0.01 & 0.00 & 0.01 & 0.00 & 0.01 & 0.02  & 0.89  \\
\bottomrule
\end{tabular}
\caption{Results for the loss weight $\lambda_{\text{safe}}$ under different attacks on the SST2 task with the Llama model.}
\label{tab:lambda_safe}
\end{table*}

\end{document}